\documentclass[lettersize,journal]{IEEEtran}
\usepackage{amsmath,amssymb}

\usepackage{subcaption}
\usepackage{bm}
\usepackage{graphicx,graphics,color,psfrag}
\usepackage{cite,balance}
\usepackage{caption}
\allowdisplaybreaks
\usepackage{algorithm}
\usepackage{algorithmic}
\usepackage{accents}
\usepackage{amsthm}
\usepackage{url}
\usepackage[english]{babel}
\usepackage{multirow}
\usepackage{enumerate}
\usepackage{cases}
\usepackage{stfloats}
\usepackage{dsfont}
\usepackage{color,soul}
\usepackage{amsfonts}
\usepackage{cite,graphicx,amsmath,amssymb}
\usepackage{fancyhdr}
\usepackage{hhline}
\usepackage{graphicx,graphics}
\usepackage{array,color}
\usepackage{mathtools}
\usepackage{amsmath}

\newtheorem{lemma}{\emph{\underline{Lemma}}}
\newtheorem{corollary}{\emph{\underline{Corollary}}}

\newtheorem{proposition}{\emph{\underline{Proposition}}}

\newtheorem{example}{\bf \emph{\underline{Example}}}
\newtheorem{remark}{\bf \emph{\underline{Remark}}}

\def\({\left(}
\def\){\right)}

\def\b0{{\mathbf{0}}}

\newcommand{\diag}{\mathrm{diag}}

\graphicspath{{./Figs/}}

\begin{document}
	\captionsetup[figure]{name={Fig.}} 
	
	\title{Near-field Physical Layer Security: Robust Beamforming under Location Uncertainty}
	\author{Chao  Zhou, 
		Changsheng~You,~
		Cong Zhou,  Chengwen Xing,~
		and Jianhua Zhang
		\thanks{
			Chao Zhou, Changsheng You and Cong Zhou are with the Department of Electronic and Electrical Engineering, Southern University of Science and Technology (SUSTech), Shenzhen
			518055, China (e-mail: zhouchao2024@mail.sustech.edu.cn, youcs@sustech.edu.cn, and zhoucong@stu.hit.edu.cn). 
		Chengwen Xing is with the School
        of Information and Electronics, Beijing Institute of Technology, Beijing 100081, China (e-mail: xingchengwen@gmail.com). Jianhua Zhang is with State Key Laboratory of Networking and Switching Technology, Beijing University of Posts and Telecommunications, Beijing 100876, China (jhzhang@bupt.edu.cn)
            \emph{(Corresponding author: Changsheng You.)}  
		}\vspace{-14pt}} 
	
	\maketitle
	\begin{abstract}
		In this paper, we study {\it robust} beamforming design for {\it near-field} physical-layer-security (PLS) systems, where a base station (BS) equipped with an extremely large-scale array (XL-array) serves multiple near-field legitimate users (Bobs) in the presence of multiple near-field eavesdroppers (Eves). Unlike existing works that mostly assume perfect channel state information (CSI) or location information of Eves, we consider a more practical and challenging scenario in this paper, where the locations of Bobs are perfectly known, while only {\it imperfect location information} of Eves is available at the BS.
		We first formulate a robust optimization problem to maximize the sum-rate of Bobs while guaranteeing a worst-case limit on the eavesdropping rate under location uncertainty. By transforming Cartesian position errors into the polar domain, we reveal an important near-field {\it angular-error amplification effect}, i.e., under the same location error, the closer the Eve, the larger the angle error, which severely degrades the performance of conventional robust beamforming methods based on imperfect channel state information. To address this issue, we first establish the conditions for which the first-order Taylor approximation of the near-field channel steering vector under location uncertainty is largely accurate.
		Then, we propose a {\it two-stage} robust beamforming method, which first partitions the uncertainty region into multiple fan-shaped sub-regions, followed by the second stage to formulate and solve a refined linear-matrix-inequality (LMI)-based robust beamforming optimization problem. In addition, the proposed method is further extended to scenarios with multiple Bobs and multiple Eves. Finally, numerical results validate that the proposed method achieves a superior trade-off between rate performance and secrecy robustness, hence significantly outperforming existing benchmarks under Eve location uncertainty.
	\end{abstract}
    \vspace{-6pt}
	\section{Introduction}
	Extremely large-scale arrays (XL-arrays) have emerged as a promising technology to enhance the spectral efficiency and spatial resolution of next-generation wireless networks~\cite{YouNGAT}. In particular, the greatly enlarged array aperture of XL-arrays fundamentally alters wireless propagation environment, transitioning from traditional far-field planar wavefront to the near-field spherical one~\cite{Cui2022CE,zhang2023channel,miao2023sub,yuan2022spatial}. 
	This property enables flexible beamforming control in the angle and range domains, which enhances the performance of various wireless applications, such as physical layer security (PLS)~\cite{LiuyuanweiPLS}, integrated sensing and communications (ISAC)~\cite{Cong_ISAC}, and next generation multiple access (NGMA)~\cite{DingNGMA}, among others. 

	In this paper, we study robust beamforming design for near-field PLS systems under imperfect location information of eavesdroppers (Eves). Specifically, we show that the performance of spotlight focusing beams in near-field systems is highly sensitive to the angle errors, thus making conventional robust beamforming designs based on imperfect channel state information (CSI) ineffective.  To tackle this issue, we propose an efficient two-stage robust near-field beamforming method by using the uncertainty region partitioning and refined general sign-definiteness (GSD) techniques.
    \vspace{-6pt}
	\subsection{Prior Works}
	\subsubsection{Far-field PLS} 
	Secrecy capacity was first introduced in~\cite{wyner1975wire}, which defines the theoretical upper bound of achievable secrecy rate. 
	For PLS systems, the spatial diversity enabled by multiple antennas allows the base station (BS) to steer directional beams toward legitimate users (Bobs) while suppressing information leakage at potential Eves~\cite{ZhengPLS}. This is achieved by customized beamforming designs and power allocation based on the acquired CSI~\cite{cwxTSP2020}.
	Moreover, for scenarios with high channel correlation between Bobs and Eves, it becomes difficult to spatially distinguish between them, making the beamforming and artificial noise (AN) ineffective.
    This issue, however, can be effectively addressed by e.g., deploying intelligent reflecting surfaces (IRSs) to dynamically reshape the channel correlation for improving secrecy performance~\cite{ZhengPLS_RIS}.

	However, in practice, obtaining accurate CSI or location information of Eves is highly challenging due to their non-cooperative nature.  To address this issue, robust beamforming design was studied to optimize secrecy performance under the worst-case or probabilistic uncertainty models, thereby ensuring reliable secrecy performance even in the presence of incomplete or imperfect knowledge of Eves~\cite{GuiZhou_Robust}. For instance, the authors in~\cite{Lin2021Robust} investigated robust beamforming in satellite–terrestrial integrated networks under imperfect angle-of-departure (AoD) information, by using uniform angular sampling to model the AoD uncertainty.
	Furthermore, this AoD uncertainty set was modeled as a norm-bounded matrix in~\cite{li2023robust} to facilitate robust beamforming designs. For scenarios with imperfect CSI, existing works (e.g.,\cite{Ng_DWK,GuiZhou_Robust}) typically assumed that the CSI error lies within a bounded uncertainty set. Based on this assumption, efficient techniques such as S-Procedure and the GSD lemma~\cite{GuiZhou_Robust} were used to guarantee the worst-case performance.

	\subsubsection{Near-field PLS}
	In contrast to far-field systems, near-field spherical wavefronts introduce an additional range-domain degree-of-freedom (DoF) and enable \emph{beam-focusing} in the range domain~\cite{zhang2024new}. This capability significantly enhances secrecy performance, particularly when Eves are located at the same angle but different distances with Bobs~\cite{LiuyuanweiPLS}.
	Exploiting this favorable property, the authors in~\cite{zhang2024PLS} analytically examined the effectiveness of AN in near-field systems and proposed a low-complexity yet efficient beamforming design tailored for near-field PLS. Building upon these insights, recent studies have investigated the performance of near-field PLS systems using low-complexity analog beamforming schemes that capitalize on the beam-focusing characteristics~\cite{NFC_PLSChen,NFWideband_PLS,liu2025physical}. For example, the authors in~\cite{NFC_PLSChen} demonstrated the potential of fully analog secure beamforming in enhancing PLS performance.
	Moreover, the analog beamforming scheme was extended in~\cite{NFWideband_PLS} and~\cite{liu2025physical} to more complex near-field wideband systems and mixed near-field and far-field communication systems, respectively. These studies revealed that the joint design of power allocation and analog beamforming is capable of achieving satisfactory secrecy performance.

    In view of the above works, near-field PLS performance typically relies on perfect CSI of both Bobs and Eves, which, however, may not be achieved in practice due to imperfect/outdated CSI. 
	To tackle this issue, the authors in~\cite{chen2025robust} considered imperfect CSI for near-field Eves, by modeling this uncertainty using a norm-bounded manner to enable robust beamforming. This robust beamforming design was further extended in ~\cite{RobustNFCISAC} to near-field secure ISAC systems. On the other hand, since CSI acquisition of near-field Eve is practically challenging, a location-aware near-field beamforming design was recently proposed~\cite{Locationmag}, which exploits geometric relationships among transceivers to enhance beamforming accuracy and robustness, especially under imperfect or unknown CSI conditions. In addition,  the authors in~\cite{NFCrobust} investigated robust beamforming for near-field systems under imperfect user location information, where an approximated CSI error bound was obtained to design robust beamforming.
	Despite these advancements, there still exist several issues in robust beamforming design for near-field PLS systems. 
\begin{itemize}
	\item First, the non-cooperative nature of Eves makes it difficult to acquire CSI of Eves in practice, which may exacerbate CSI estimation errors and undermine the effectiveness of existing beamforming designs.
	\item Second, the spotlight focusing beams in near-field systems are highly sensitive to the location errors of both Bobs and Eves. However, it is unknown how the angle and/or range errors affect the near-field rate and secrecy performance.
    \item Third, existing robust beamforming methods, especially those based on CSI uncertainty, mostly used linear matrix inequality (LMI) formulations. However, in near-field systems with a huge number of antennas, the high-dimensional LMIs incur prohibitively high computational complexity, making it unaffordable in practice.
    \end{itemize}
	\vspace{-10pt}
	\subsection{Contributions}
	Motivated by the above,  we study in this paper a robust beamforming design for near-field PLS systems as shown in Fig.~\ref{Fig:SystemModel}, where a BS equipped with an XL-array transmits confidential signals to multiple near-field Bobs in the presence of multiple near-field Eves. Unlike existing near-field PLS studies that mostly assume perfect CSI of Eves available at the BS, we consider a more practical and challenging scenario where only imperfect location information of Eves is available, while the location information of Bobs is assumed to be perfectly known at the BS based on e.g., location-information sharing between them.  To ensure secure communications under such location uncertainty, we propose a robust beamforming method under location uncertainty of Eves. The main contributions of this paper are summarized as follows.
	
	\begin{itemize}
		\item First, we formulate a robust beamforming optimization problem to maximize the achievable sum-rate of multiple Bobs, while ensuring the worst-case eavesdropping rate at each Eve below a prescribed threshold under location uncertainty. By transforming the position uncertainty from the Cartesian coordinate system to the polar coordinate system, we reveal an important near-field \emph{angular-error amplification effect}, i.e., under the same location error, the closer the Eve, the larger the angle error, which severely degrades the performance of conventional robust beamforming methods.
		
		\item Second, for the special case with one Bob and one Eve, we analyze the accuracy of the first-order Taylor approximation for the near-field channel steering vector (CSV) under location uncertainty and establish the conditions under which the first-order Taylor approximation is largely accurate. Based on this, we propose a new two-stage robust beamforming design. In the first stage, the spherical location uncertainty region is partitioned into multiple fan-shaped sub-regions such that the Taylor approximation is locally accurate within each sub-region. Then, in the second stage, by applying the first-order Taylor expansion around a surrogate location within each sub-region, we formulate and solve a refined LMI-based optimization problem that converts the intractable secrecy constraints into tractable forms with low computational~complexity.
		
		\item Finally, we extend the proposed method to the more general scenario involving multiple Bobs and multiple Eves for 
		maximizing the achievable sum-rate, while ensuring each Bob meeting its worst-case secrecy rate constraint across all Eves. 
		In addition, we discuss the general scenarios with location uncertainty of Bobs and multi-path
		channels. Numerical results demonstrate that the proposed method achieves a favorable rate performance and robust secrecy, even when the angle errors of Eves caused by the near-field angular-error amplification effect is significant.
	\end{itemize}

	\section{System Model and Problem Formulation}\label{Sec2:SysModel}
	We consider an XL-array enabled secure communication system as shown in Fig.~\ref{Fig:SystemModel}, where an $N$-antenna BS serves $K$ single-antenna Bobs in the downlink, which are denoted by $\mathcal{K}\triangleq \{1,2,\ldots,K\}$. 
	Meanwhile, $M$ single-antenna Eves, denoted by  $\mathcal{M}\triangleq \{1,2,\ldots,M\}$, are located near the BS for intercepting legitimate communications at short distance. 
	Without loss of generality, the XL-array is placed at the $y$-axis, with its array center positioned at $\mathbf{u}_{0} = [0,0]^T$. Given half-wavelength inter-antenna spacing, the position of  antenna $n$ is given by $\mathbf{u}_{n} = [0,u_{n}]^T$, where 
	$	u_{n} = \frac{(2n - N - 1)}{2} d,~\forall n \in \mathcal{N}$,
	with $d \triangleq \lambda/2$  and $\lambda$ denoting the carrier frequency.
	\begin{figure}[t]
		\centering
		\includegraphics[width=0.3\textwidth]{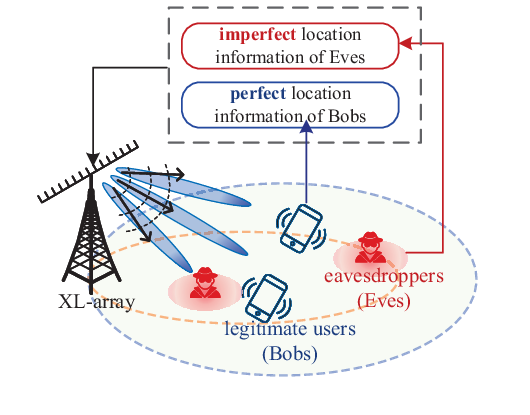}
		\caption{An XL-array enabled near-field PLS system.} \label{Fig:SystemModel}
		\vspace{-15pt}
	\end{figure}

	\subsection{Channel Model}
	The propagation channels associated with the XL-array are modeled based on spherical wavefronts. First, consider an arbitrary Bob $k$, located at $\mathbf{q}_{{\rm B},k} = [x_{{\rm B},k},y_{{\rm B},k}]^T$ in the Cartesian coordinate. We denote by $ \mathbf{h}_{{\rm B},k}^{H} \in \mathbb{C}^{1 \times N}$  its channel from the BS, modeled as
	\begin{align}\label{Exp:Bobchannel}
		\mathbf{h}_{{\rm B},k}^{H} \!=\! \sqrt{N}h_{{\rm B},k} \mathbf{a}^{H}(\!\mathbf{q}_{{\rm B},k}\!)
		\!+\!\sqrt{\!\frac{N}{L_{{\rm B},k}}}\!\sum_{\ell=1}^{L_{{\rm B},k}}\! h_{{\rm B},k,\ell} \mathbf{a}^{H}(\!\mathbf{q}_{{\rm B},k,\ell}\!),
	\end{align}
	which consists of one line-of-sight (LoS) path and $L_{{\rm B},k} $ non-LoS (NLoS) paths. Herein, $h_{{\rm B},k} = \frac{\sqrt{h_0}}{r_{{\rm B},k}} e^{-\jmath \frac{2\pi}{\lambda}r_{{\rm B},k}}$  denotes the complex-valued channel gain where $h_0$ is the reference channel gain at a range of 1  meter (m), and $r_{{\rm B},k} = \| \mathbf{q}_{{\rm B},k} - \mathbf{u}_{0} \|_{2}$ is the distance between Bob $k$ and the XL-array center. Additionally, $h_{{\rm B},k,\ell}$ and $\mathbf{q}_{{\rm B},k,\ell}$ denote the complex-valued channel gain and the scatterer location of path $\ell$, respectively. We consider high-frequency band scenarios, where NLoS paths exhibit negligible power due to severe path-loss and shadowing. 
    As such, the channel of Bob $k$ in~\eqref{Exp:Bobchannel}  can be approximated by its LoS component, i.e., $\mathbf{h}_{{\rm B},k}^{H} \approx \sqrt{N}h_{{\rm B},k} \mathbf{a}^{H}(\mathbf{q}_{{\rm B},k})$,\footnote{In multi-path scenarios, the NLoS component can be upper-bounded by its maximum power, upon which the proposed method can be extended, with details provided in {\bf Remark~\ref{Dis:Multipath}}.} where the near-field CSV, $\mathbf{a}(\mathbf{q}_{{\rm B},k})$, is given by
	\begin{align}\label{Exp:Gensteering}
        [\mathbf{a}(\mathbf{q}_{{\rm B},k})]_{n} = \frac{1}{\sqrt{N}} e^{-\jmath \frac{2\pi}{\lambda} \big(\|\mathbf{q}_{{\rm B},k} - \mathbf{u}_{n} \|_{2} - \|\mathbf{q}_{{\rm B},k} - \mathbf{u}_{0} \|_{2} \big) }.
	\end{align}
	Herein, $\|\mathbf{q}_{{\rm B},k} - \mathbf{u}_{n} \|_{2}$ represents the distance between Bob $k$ and the $n$-th XL-array antenna, which can be obtained based on Fresnel approximation
	\begin{align}
		&\|\mathbf{q}_{{\rm B},k} - \mathbf{u}_{n} \|_{2}  = \sqrt{r_{{\rm B},k}^{2} + u_{n}^{2} - 2 r_{{\rm B},k} u_{n} \sin\theta_{{\rm B},k}} \nonumber \\
		&\overset{(a)}{\approx} r_{{\rm B},k} -u_{n} \sin\theta_{{\rm B},k} + \frac{u_{n}^2 \cos^2\theta_{{\rm B},k}}{2r_{{\rm B},k}}, 
	\end{align}
	where $r_{{\rm B},k}$ and $  \theta_{{\rm B},k} $ are respectively the range and (physical) angle between the BS and Bob $k$, given by
	\begin{align}
		r_{{\rm B},k} = \sqrt{x_{{\rm B},k}^{2} + y_{{\rm B},k}^{2}},~ \theta_{{\rm B},k} = \arctan\left(\frac{y_{{\rm B},k}}{x_{{\rm B},k}}\right).
	\end{align}
	As such, the near-field CSV in~\eqref{Exp:Gensteering} can also be expressed as the following function of  $ r_{{\rm B},k} $ and $ \theta_{{\rm B},k} $
	\begin{align}\label{Exp:NF_steeringvec}
        [\mathbf{a}(\theta_{{\rm B},k},r_{{\rm B},k})]_{n} = \frac{1}{\sqrt{N}} e^{\jmath \frac{2\pi}{\lambda} \big(u_{n} \sin \theta_{{\rm B},k} - \frac{u_{n}^2 \cos^2\theta_{{\rm B},k}}{2 r_{{\rm B},k}} \big) }.
	\end{align}
	
	Similarly, the channel from the BS to the $m$-th Eve, denoted by $\mathbf{h}_{{\rm E},m}^{H} \in \mathbb{C}^{1 \times N} $, can be modeled as 
	\begin{align}
		\mathbf{h}_{{\rm E},m}^{H}  = \sqrt{N} h_{{\rm E},m}\mathbf{a}^{H}(\mathbf{q}_{{\rm E},m}),
	\end{align} 
	where $h_{{\rm E},m} = \frac{\sqrt{h_0}}{r_{{\rm E},m}} e^{-\jmath \frac{2\pi}{\lambda}r_{{\rm E},m}}$ denotes the complex-valued channel gain and $ \mathbf{q}_{{\rm E},m} = [x_{{\rm E},m}, y_{{\rm E},m}]^T $ is the location of Eve. By using $ r_{{\rm E},m} = \sqrt{x_{{\rm E},m}^{2} + y_{{\rm E},m}^{2}}$ and $\theta_{{\rm E},m} = \arctan\left({y_{{\rm E},m}}/{x_{{\rm E},m}}\right) $, the near-field CSV of Eve $m$ can also be denoted as $\mathbf{a}(\theta_{{\rm E},m},r_{{\rm E},m})$. 
	\vspace*{-20pt}
	\subsection{Location Error Model}
	Accurate CSI of both Bobs and Eves is crucial for efficient beamforming designs, which primarily depends on the \emph{location} information in near-field LoS-dominant scenarios. In this paper, we assume that the BS has perfect location information of Bobs, since Bobs usually can cooperate with the BS for CSI acquisition and location-information sharing.\footnote{The locations of Bobs can be practically acquired using existing near-field localization methods~(see, e.g., \cite{dai2025tutorial}), as well as classical methods such as those based on time-of-arrival, received signal strength, radio fingerprinting, etc. Moreover, in {\bf Remark~\ref{Dis:BobLoc_error}}, we further discuss the extension to scenarios with location errors of Bobs.}
	Nevertheless, due to the non-cooperative nature of Eves, obtaining their (perfect) CSI is challenging, making it difficult to design beamforming vectors to ensure PLS. Instead of directly acquiring the CSI of Eves, an alternative approach is by leveraging their location information, which is generally easier to obtain through e.g., non-cooperative localization methods~\cite{Sun_Evelocation,Liu_Evelocation}.
	For example, locations of Eves can be estimated by monitoring local oscillator power leakage from the radio frequency front-ends of their receivers~\cite{Sun_Evelocation}. However, this method is vulnerable to environmental interference and received noise, resulting in inaccurate location estimation in general~\cite{Sun_Evelocation,Liu_Evelocation}. As such, only \emph{imperfect} location information of Eves is assumed in this paper. Specifically, the location of Eve $m$ is modeled as
	\begin{align}\label{Exp:Loc_Eves}
		\mathbf{q}_{{\rm E},m} = \hat{\mathbf{q}}_{{\rm E},m} + \Delta \mathbf{q}_{{\rm E},m},~\forall m\in\mathcal{M},
	\end{align}
	where $\hat{\mathbf{q}}_{{\rm E},m} = [\hat{x}_{{\rm E},m}, \hat{y}_{{\rm E},m}]^T$ 
	is the estimated location of Eve $m$ and $\Delta  \mathbf{q}_{{\rm E},m} = [\Delta x_{{\rm E},m}, \Delta y_{{\rm E},m}]^T$ represents the location error in the Cartesian coordinate. 
	By employing non-cooperative localization methods, the location error vector can be modeled as a two-dimensional Gaussian distribution\footnote{In non-cooperative localization scenarios, such as those relying on received signal strength (RSS) or time of arrival (ToA) measurements, the location estimation error is the cumulative effect of several factors, including e.g., thermal noise and environmental scattering. As a result, the combination of these errors generally follows a Gaussian distribution~\cite{Sun_Evelocation, Liu_Evelocation}. In scenarios where location information of Eves is unavailable, an AN-aided jamming strategy serves as a viable approach~\cite{zhang2024PLS}. This approach projects AN signals into the null space of the legitimate communication channel, effectively suppressing potential Eves regardless of their locations.}~\cite{Sun_Evelocation,Liu_Evelocation}
	\begin{align}
		\Delta \mathbf{q}_{{\rm E},m} & \sim \mathcal{N}(\boldsymbol{\mu}_{{\rm E},m}, \boldsymbol{\Sigma}_{{\rm E},m} ),
	\end{align}
	where $\boldsymbol{\mu}_{{\rm E},m} = [\mu_{x,m},\mu_{y,m}]^T$ with $\mu_{x,m}$ and $\mu_{y,m}$ denoting the means of  $\Delta x_{{\rm E},m}$ and $\Delta y_{{\rm E},m}$, respectively. In addition,  $ \boldsymbol{\Sigma}_{{\rm E},m} $ is the covariance matrix, given by
	\begin{align}	
		\boldsymbol{\Sigma}_{{\rm E},m} = 
		\begin{bmatrix}
			\sigma_{x,m}^{2} & \rho \sigma_{x,m} \sigma_{y,m} \\
			\rho \sigma_{x,m} \sigma_{y,m} & \sigma_{y,m}^{2}
		\end{bmatrix},
	\end{align}
	where $\sigma_{x,m}^{2}$ and $\sigma_{y,m}^{2}$ are the variances of $\Delta x_{{\rm E},m}$ and $\Delta y_{{\rm E},m}$, respectively, and $\rho$ represents the correlation coefficient between these two variables. Similar to \cite{Liu_Evelocation}, we consider the case where the location errors $\Delta x_{{\rm E},m}$ and $\Delta y_{{\rm E},m}$ are independent Gaussian variables with zero mean and variances of $\sigma_{{\rm c},m}^{2}$.\footnote{For the scenario where $\mu_{x,m}\! \neq\! \mu_{y,m}\!\neq 0,~\sigma_{x,m}^{2}\! \neq\! \sigma_{y,m}^{2}$, and $\rho\! \neq\! 0 $, the uncertainty region for Eve $m$ is given by $(\mathbf{q}_{{\rm E},m}-\boldsymbol{\mu}_{{\rm E},m})^{T} \boldsymbol{\Sigma}_{{\rm E},m}^{-1}  (\mathbf{q}_{{\rm E},m}- \boldsymbol{\mu}_{{\rm E},m}) \le \chi_{1-\alpha}^{2}$. Based on this elliptical uncertainty region, the proposed robust beamforming method can be directly extended to this generalized case.} As such, the location error vector is distributed as
	\begin{align}
		\Delta \mathbf{q}_{{\rm E},m} & \sim \mathcal{N}\Bigg(
		\begin{bmatrix}
			0  \\
			0
		\end{bmatrix},
		\begin{bmatrix}\label{Eq:ErrorGua}
			\sigma_{{\rm c},m}^{2} & 0 \\
			0 & \sigma_{{\rm c},m}^{2}
		\end{bmatrix}
		\Bigg).	
	\end{align} 
	
	\subsection{Signal Model}
	Let $s_{k} \sim \mathcal{CN}(0,1)$ and $\mathbf{w}_{k} \in \mathbb{C}^{N \times 1}$, for $k\in \mathcal{K}$ denote the transmitted signal and beamforming vector of the BS to each Bob $k$. As such, the received signal at  Bob $k$ is given by 
	\begin{align}
		\underline{y}_{{\rm B},k} = \mathbf{h}_{{\rm B},k}^{H} \mathbf{w}_{k} s_{k} + \sum_{i=1,i\neq k}^{K} \mathbf{h}_{{\rm B},k}^{H} \mathbf{w}_{i} s_{i} + z_{{\rm B},k}, \forall k \in \mathcal{K},
	\end{align}
	where $z_{{\rm B},k} \sim \mathcal{CN}(0,\sigma^{2})$ is the received additive white Gaussian noise (AWGN). Accordingly, the achievable rate of Bob $k$ in bits/second/Hertz (bps/Hz) is 
	\begin{align}
		\!\!\!\!\!{R}_{{\rm B},k} = \log_{2} \Big( 1 +\frac{| \mathbf{h}_{{\rm B},k}^{H} \mathbf{w}_{k} |^2}{ \sum_{i=1,i\neq k}^{K} | \mathbf{h}_{{\rm B},k}^{H} \mathbf{w}_{i}   |^{2}  + \sigma^{2} } \Big),\forall k \in \mathcal{K}.\!
	\end{align} 
	Similarly, for Eves aiming at intercepting transmitted confidential information, the received signal at Eve $m$ is 
	\begin{equation}
		\underline{y}_{{\rm E},m} = \sum_{k=1}^{K} \mathbf{h}_{{\rm E},m}^{H} \mathbf{w}_{k} s_{k} + z_{{\rm E},m}, \forall m \in \mathcal{M},
	\end{equation}
	where $z_{{\rm E},m} \sim \mathcal{CN}(0,\sigma^{2})$ is the received AWGN at Eve $m$. We consider the scenario where each Eve is capable of canceling all multi-user interferences before decoding legitimate information without cooperation~\cite{Ng_DWK,zhang2024PLS}, which poses a practical challenge for secure communications.\footnote{In this work, we consider the most challenging worst-case scenarios to characterize the lower bound of system performance under the assumption of perfect CSI of all Bobs, while the proposed method is applicable to the case where residual interference arises due to imperfect CSI or incomplete cancellation, as well as inherent multi-user interference, by introducing auxiliary variables as in~\cite{GuiZhou_Robust}.}
	As such, the eavesdropping rate of Eve $m$ for wiretapping Bob $k$ is 
	\begin{align}
		R_{{\rm E},m,k} = \log_{2}\Big(1 + {|\mathbf{h}_{{\rm E},m}^{H} \mathbf{w}_{k} |^2}/{\sigma^{2} } \Big),~\forall m, k.
	\end{align}

	\vspace{-6pt}
	\subsection{Problem Formulation}
	Our target is to maximize the sum-rate of all legitimate users while ensuring the maximum information leakage to each Eve below a prescribed threshold, accounting for  location uncertainty of Eves. Specifically, to facilitate the worst-case robust beamforming design, we consider a \emph{bounded} spherical uncertainty region $\mathcal{A}_{{\rm E},m}$ (instead of original unbounded counterpart \eqref{Eq:ErrorGua}) for the location error of each Eve $m$, where $\hat{\mathbf{q}}_{{\rm E},m}$ is the uncertainty region center and its uncertainty radius is given by~\cite{altman2013statistics}  
	\begin{align}
		\Upsilon_{m}^{2} = \sigma_{{\rm c},m}^{2} \chi_{1-\alpha}^{2}(2), ~\forall m \in \mathcal{M}.
	\end{align}
	Herein, $\chi_{1-\alpha}^{2}(2)$ is the inverse cumulative of the Chi-square distribution with two DoFs and $1-\alpha$ denotes the confidence level (e.g., $\alpha = 0.05 $ corresponds to $95\%$ confidence). This region $\mathcal{A}_{{\rm E},m}$ ensures that the actual location of  Eve $m$ lies within it with likelihood over e.g., 95\%. Let $\hat{r}_{{\rm E},m} = \sqrt{ (\hat{x}_{{\rm E},m})^2 + (\hat{y}_{{\rm E},m})^2  }$ and $\hat{\theta}_{{\rm E},m} = \arctan\left(\frac{\hat{y}_{{\rm E},m}}{\hat{x}_{{\rm E},m}} \right)$ represent the estimated range and angle of the $m$-th Eve. 
	For each Eve $m$ within its spherical uncertainty region $\mathcal{A}_{{\rm E},m}$, 
	the angle and range errors, denoted by $\Delta \theta_{{\rm E},m}\triangleq \theta_{{\rm E},m}- \hat{\theta}_{{\rm E},m}$ and $\Delta r_{{\rm E},m} \triangleq r_{{\rm E},m}- \hat{r}_{{\rm E},m}$, respectively, are bounded by 
	\begin{subequations}
		\begin{align}
			&| \Delta \theta_{{\rm E},m} | \le \arcsin\Big({\Upsilon_{\rm m}}/{\hat{r}_{{\rm E},m}}\Big),\\
			&|	\Delta r_{{\rm E},m} | \le \Upsilon_{\rm m},
		\end{align}
	\end{subequations}
	which can be obtained based on the geometric relationship between the estimated position and the uncertainty boundary.
	
	As such, the \emph{robust} near-field sum-rate maximization problem under the maximum eavesdropping rate constraints can be  formulated as
	\begin{subequations}
		\begin{align}
			(\textbf{P1}):\; \max_{\{\mathbf{w}_{k}\}} \quad &\sum_{k=1}^{K}~ R_{{\rm B},k}  \nonumber \\
			\text{s.t.}   \quad & \sum_{k=1}^{K}~ \| \mathbf{w}_{k} \|_{2}^2 \le P_{\text{max}}, \label{C:Power}\\
			&\max_{\{\mathbf{q}_{{\rm E},m}\in \mathcal{A}_{{\rm E},m}\}}  R_{{\rm E},m,k} \le R_{\max},~\forall m, k,   \label{C:Secrecy}
		\end{align}
	\end{subequations}
	where  constraint~\eqref{C:Power} enforces the maximum BS transmit power $P_{\max}$ and constraint~\eqref{C:Secrecy} ensures that the worst-case eavesdropping rate of each Eve for wiretapping any Bob is no larger than a prescribed threshold $R_{\max}$ under the bounded location errors $\{\mathbf{q}_{{\rm E},m}\in \mathcal{A}_{{\rm E},m}\}$.
	
	Problem~\textbf{(P1)} is a non-convex optimization problem and thus hard to be optimally solved in general, since 1) the objective function is non-concave with respect to (w.r.t.) the beamforming vectors $\{\mathbf{w}_{k}\}$, and 2) the location uncertainty of Eves introduces \emph{infinitely many} non-convex constraints in~\eqref{C:Secrecy} due to continuous uncertainty sets. To tackle these difficulties and obtain useful insights, in the following, we first consider a single-Bob-single-Eve case for which an efficient method is proposed to solve Problem~\textbf{(P1)}, then the method is further extended to the general case.

	\section{Single-Bob-Single-Eve Scenario}\label{Sec:III}
	For the single-Bob-single-Eve case, we first introduce two conventional methods and point out their main limitations. Then,  an efficient \emph{two-stage} robust beamforming method is proposed to effectively overcome these limitations and aforementioned challenges in solving Problem~\textbf{(P1)}. For notational brevity, the indices of Bob and Eve are omitted in this section.
	
	First, for the considered case, Problem~\textbf{(P1)} reduces to\footnote{We consider the case where the location error is much smaller than the estimated distance, for which the amplitude error of channel $\mathbf{h}_{{\rm E},m}^{H}$ is very small and thus is neglected.} 
	\begin{subequations}
		\begin{align}
			(\textbf{P2}):\; \max_{\mathbf{w}} \quad &|\mathbf{h}_{\rm B}^H \mathbf{w} |^2 \nonumber \\
			\text{s.t.}   \quad &  \| \mathbf{w} \|_{2}^2 \le P_{\text{max}}, \label{C:Power_S}\\
			&	\max_{ (\theta_{\rm E},r_{\rm E}) \in \mathcal{A}_{\rm E} }  |\mathbf{a}^H(\theta_{\rm E},r_{\rm E}) \mathbf{w} |^2 \le \Gamma,  \label{C:Secrecy_S}
		\end{align}
	\end{subequations}
	where  $\Gamma=\sigma^2 (2^{R_{\max}}-1)/{N|h_{\rm E}|^2}$. This problem, however, is still difficult to solve due to the non-concave objective function and an infinite number of constraints in~\eqref{C:Secrecy_S}. 
	To address these issues, we first employ the successive convex approximation (SCA) technique to construct a concave surrogate function for ($|\mathbf{h}_{\rm B}^H \mathbf{w} |^2$), i.e., $
		g\big(\mathbf{w}|\mathbf{w}^{(j)}\big) \triangleq~| \mathbf{h}_{\rm B}^H \mathbf{w}^{(j)} |^2 
		+ 2\mathcal{R} \big\{ \big(\mathbf{w}^{(j)}\big)^H \mathbf{h}_{\rm B} \mathbf{h}_{\rm B}^H \big(\mathbf{w} -\mathbf{w}^{(j)} \big) \big\}$,
	where $\mathbf{w}^{(j)}$ is a feasible beamforming vector in the $j$-th SCA iteration. Next, for the non-convex constraint in~\eqref{C:Secrecy_S}, there are two typical solution methods in the existing literature (e.g.,~\cite{Lin2021Robust,li2023robust}), namely, 1) discretizing the continuous uncertainty set $\mathcal{A}_{\rm E}$ into a finite number of samples, and 2) exploiting the channel error bound to approximate the constraints~\eqref{C:Secrecy_S} as a single LMI~\cite{li2023robust}. The details of these two methods are elaborated as follows.
	
	\vspace{-6pt}
	\subsection{Conventional Methods}\label{Sec:III-A}
	\subsubsection{Sampling-based method}
	Consider the widely used uniform location sampling method in the spherical uncertainty region $\mathcal{A}_{\rm E}$~\cite{BehjooSampling}. Let $\bar{S}$ denote the total number of sampling points, for which the possible locations of Eve are sampled at $\bar{\mathbf{q}}_{\mathrm{E},\bar{s}}$,  $\forall \bar{s} \in \bar{\mathcal{S}}\triangleq\{1,2,\ldots,\bar{S}\}$. 
	As such, \eqref{C:Secrecy_S} can be  approximated as
	\begin{align}
		|\mathbf{a}^H({\bar{\mathbf{q}}_{{\rm E},\bar{s}}}) \mathbf{w} |^2 \le \Gamma,~\forall \bar{s} \in \bar{\mathcal{S}}.
	\end{align}
	This transformation converts constraints in~\eqref{C:Secrecy_S} into $\bar{S}$ convex quadratic constraints, hence allowing the use of standard convex optimization tools to solve this convex optimization problem.
	However, the performance of this sampling-based method critically depends on the number of sampling points.
    An insufficient number of sampling points fail to adequately traverse the continuous uncertainty set $\mathcal{A}_{\rm E}$ and thus cannot guarantee the eavesdropping rate constraint~\eqref{C:Secrecy_S} in the uncertainty set.
	On the other hand, a huge number of sampling points provides a more accurate approximation for the uncertainty set, while it incurs prohibitively high computational complexity proportional to the number of sampling points, hence making it unaffordable in practice.

	\subsubsection{Error-bound-based method}
	\begin{figure}[t]
		\centering
		\includegraphics[width=0.3\textwidth]{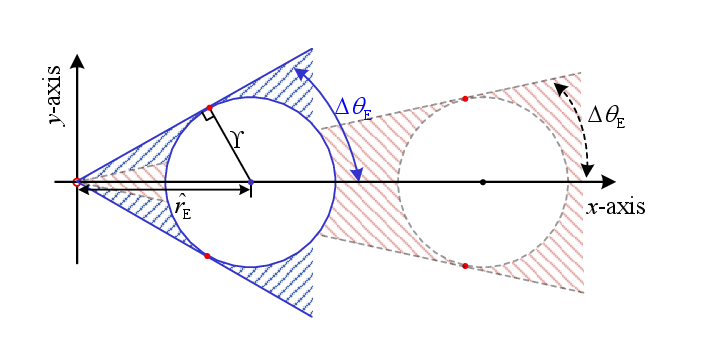}
		\caption{Schematic of near-field angular-error amplification.} \label{Fig:Angerr_amp}
		\vspace{-12pt}
	\end{figure}
	
	For the error-bound-based method, the CSV error bound is utilized to transform the infinite number of constraints into a single LMI~\cite{GuiZhou_Robust}. Specifically, the CSV of Eve at an arbitrary location in $\mathcal{A}_{\rm E}$ is approximated as follows by its  first-order Taylor expansion at location $(\hat{\theta}_{\rm E},\hat{r}_{\rm E})$~\cite{li2023robust}
	\begin{align}\label{Exp:NFV_Taylor}
		\mathbf{a}(\theta_{\rm E},r_{\rm E})&\approx \mathbf{a}(\hat{\theta}_{\rm E},\hat{r}_{\rm E})\nonumber\\
		&~~~~+  \triangledown_{\theta}\mathbf{a}|_{(\hat{\theta}_{\rm E},\hat{r}_{\rm E})}\Delta\theta_{\rm E}
		+ \triangledown_{r}\mathbf{a}|_{(\hat{\theta}_{\rm E},\hat{r}_{\rm E})}\Delta r_{\rm E},
	\end{align}
	where $\triangledown_{\theta}\mathbf{a}$ and $ \triangledown_{r}\mathbf{a} $ are gradients of $\mathbf{a}(\theta,r)$ w.r.t. $\theta$ and ${r}$, respectively (detailed in Section~\ref{Sec:III-B}).
	Based on the above,  the \emph{maximum} CSV error  can be upper-bounded as
	\begin{align}\label{Exp:TaylorBoundofvec}
		&\max_{ (\theta_{\rm E},r_{\rm E}) \in \mathcal{A}_{\rm E} }  \big\|\mathbf{a}(\theta_{\rm E},r_{\rm E})- \mathbf{a}(\hat{\theta}_{\rm E},\hat{r}_{\rm E})\big\|_2 \triangleq \varepsilon^{(\rm ub)} \nonumber \\ 
		= & \max_{ (\theta_{\rm E},r_{\rm E}) \in \mathcal{A}_{\rm E} } \|  \triangledown_{\theta}\mathbf{a}|_{(\hat{\theta}_{\rm E},\hat{r}_{\rm E})}\Delta\theta_{\rm E}
		+ \triangledown_{r}\mathbf{a}|_{(\hat{\theta}_{\rm E},\hat{r}_{\rm E})}\Delta r_{\rm E} \|_{2} \nonumber\\
		\le  &\underbrace{\| \triangledown_{\theta}\mathbf{a}|_{(\hat{\theta}_{\rm E},\hat{r}_{\rm E})}\arcsin(\Upsilon/\hat{r}_{\rm E}) \|_2}_{\varepsilon_{\theta,{\rm Tayl}}^{(\rm ub)}} + \underbrace{\|\triangledown_{r}\mathbf{a}|_{(\hat{\theta}_{\rm E},\hat{r}_{\rm E})} \Upsilon \|_2}_{\varepsilon_{r,{\rm Tayl}}^{(\rm ub)} }  
		\triangleq  \varepsilon_{\rm Tayl}^{(\rm ub)}, 
	\end{align} 
	where $\varepsilon^{(\rm ub)} \le  \varepsilon_{\rm Tayl}^{(\rm ub)} $ is due to the triangle inequality.
	Based on~\eqref{Exp:TaylorBoundofvec}, constraint~\eqref{C:Secrecy_S} can be rewritten as the following LMI by using GSD~\cite{GuiZhou_Robust}\footnote{Note that this reformulation can be obtained based on the GSD lemma in~\cite{GuiZhou_Robust}, with detailed derivation provided in Section~\ref{Sec:III-B2}.}
	\begin{align}\label{Exp:Con_LMI}
		\mathbf{G}_{\rm E}\triangleq \left[
		\begin{array}{c@{~}c@{~}c@{}}
			{\Gamma} - \hat{\lambda}_{\rm E} &  \mathbf{a}(\hat{\theta}_{\rm E},\hat{r}_{\rm E})^H\mathbf{w} & \mathbf{0}_{(1\times N)}  \\
			\mathbf{w}^H\mathbf{a}(\hat{\theta}_{\rm E},\hat{r}_{\rm E})  &     1     & \varepsilon_{\rm Tayl}^{(\rm ub)} \mathbf{w}^H \\
			\mathbf{0}_{(N \times 1)} &\varepsilon_{\rm Tayl}^{(\rm ub)} \mathbf{w}  & \hat{\lambda}_{\rm E} \mathbf{I}_{(N\times N)} 
		\end{array}
		\right] \! \!\succeq\! \mathbf{0},
	\end{align}
	where $\hat{\lambda}_{\rm E}\ge 0 $ is an auxiliary variable.

	For the considered location error model in~\eqref{Exp:Loc_Eves}, given the location error parameter $\Upsilon$ (or $\sigma_{{\rm c}}$), the angle error (i.e., $ \Delta \theta_{\rm E} = \arcsin(\Upsilon/\hat{r}_{\rm E})$) may become very large when the estimated range of Eve is relatively small, which is termed as the near-field \emph{angular-error amplification} effect, as shown in Fig.~\ref{Fig:Angerr_amp}. This effect imposes two limitations on the performance of the error-bound-based method. \emph{First}, the first-order Taylor approximation of $\mathbf{a}(\theta_{\rm E}, r_{\rm E})$ at point ($\hat{\theta}_{\rm E},\hat{r}_{\rm E}$) (see~\eqref{Exp:NFV_Taylor}) may not be accurate for all potential locations of Eve $(\theta_{\rm E}, r_{\rm E})\in \mathcal{A}_{\rm E}$, and the approximation error becomes significant when the angle error ($\Delta \theta_{\rm E}$) and/or range error ($\Delta r_{\rm E}$) are relatively large~\cite{er1994new}, as illustrated in Fig.~\ref{Fig4:CSVofTaylor}(a). \emph{Second}, the performance of the error-bound-based method is constrained by the magnitude of the error bound (i.e., $\varepsilon_{\rm Tayl}^{(\rm ub)}$), which tends to increase with the decreasing range $\hat{r}_{\rm E}$ due to the angle-error amplification effect, hence resulting in a \emph{conservative} beamforming design given secrecy requirements, as presented below.
	
	\begin{figure}[t]
	\centering
	\begin{subfigure}{0.49\linewidth}
		\centering
		\includegraphics[width=1\linewidth]{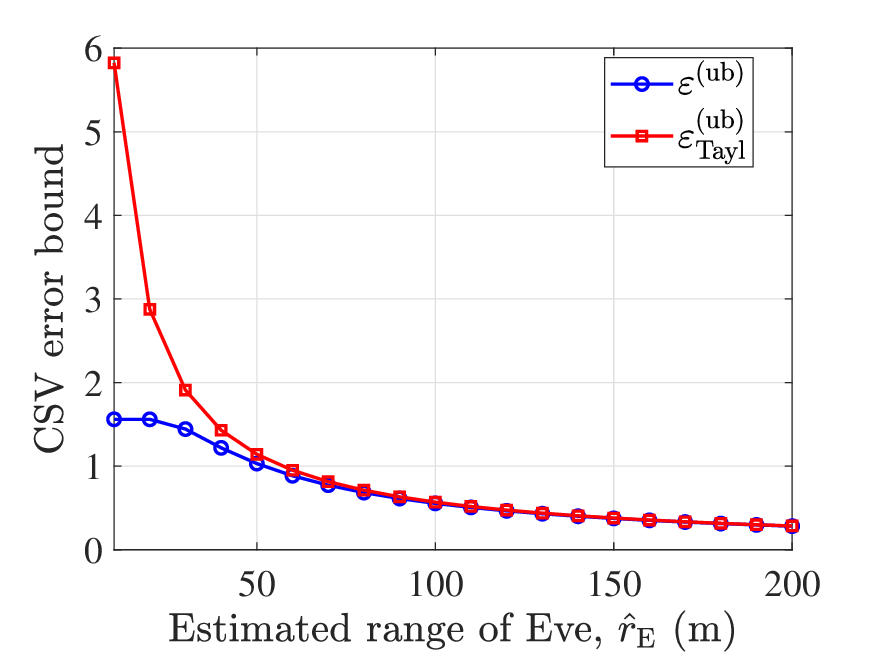}
		\caption{CSV error bound versus $\hat{r}_{\rm E}$ with $\sigma_{c} = 0.1$.}
		\label{Fig:Sec3_TaylBou}
	\end{subfigure}
	\begin{subfigure}{0.49\linewidth}
		\vspace{0em}
		\centering
		\includegraphics[width=1\linewidth]{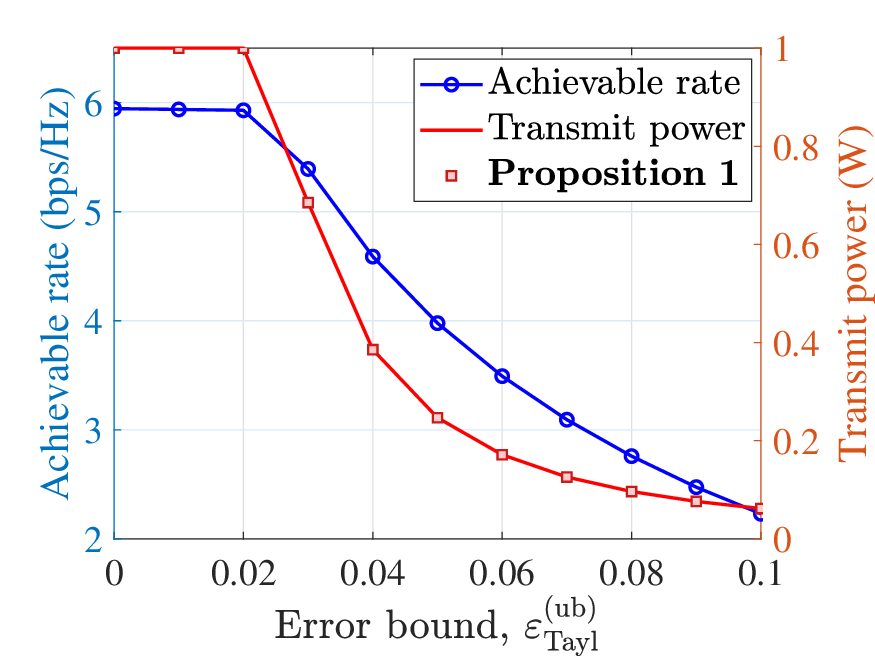}
		\caption{Achievable rate and transmit power versus error bound.}
		\label{Fig:Sec3_Boundrate_and_power}
	\end{subfigure}
	\caption{CSV error bound, and corresponding achievable rate and transmit power.}
	\vspace{-20pt}
	\label{Fig4:CSVofTaylor}
	\end{figure}

	\begin{proposition}\label{Pro:Power_limit}
		\rm
		For the  error-bound-based method, given $\Gamma$, the error bound $\varepsilon_{\rm Tayl}^{(\rm ub)}$ in~\eqref{Exp:TaylorBoundofvec}, and estimated Eve location $(\hat{\theta}_{\rm E},\hat{r}_{\rm E})$,  the beamforming vector $\mathbf{w}$ should satisfy $\|\mathbf{w}\|_{2}^2 \le \Gamma/(\varepsilon_{\rm Tayl}^{(\rm ub)})^2$ under the LMI constraint \eqref{Exp:Con_LMI}.
	\end{proposition}
	\vspace{-5pt}
	\begin{proof}
		Please refer to Appendix~\ref{App:PowerUB}.
	\end{proof}

	As shown in Fig.~\ref{Fig4:CSVofTaylor}(b), given the estimated location of Eve at $(10, 0)$ m, the optimized transmit power of the error-bound-based method monotonically decreases with the CSV error bound, as stated in {\bf{Proposition 1}}. Notably, even for a relatively small CSV error bound (e.g., 0.1, which is much smaller than that in Fig.~\ref{Fig4:CSVofTaylor}(a)), the allocated BS transmit power approaches zero, resulting in a substantially degraded achievable rate.

	\vspace{-6pt}
	\subsection{Proposed Method}\label{Sec:III-B}
	In this subsection, we propose an efficient near-field robust secure beamforming method that achieves superior PLS performance with low complexity, hence outperforming both conventional methods in Section~\ref{Sec:III-A}. Essentially, our proposed method first determines a valid region where the first-order Taylor approximation of CSV is accurate enough, based on which the uncertainty region is divided into a finite and small number of sub-regions for designing robust beamforming.

	\subsubsection{Conditions for accurate first-order Taylor approximation}\label{Sec:III-B1}
	To overcome the limitations of the error-bound-based method in Section~\ref{Sec:III-A}, we first characterize the CSV errors caused by angle and range deviations, and then establish the conditions under which the first-order Taylor approximation is accurate.
	
	Specifically, in the spherical uncertainty region $\mathcal{A}_{\rm E}$, the CSV error of Eve is 
	\begin{align}
		&\|\mathbf{a}({\theta}_{\rm E},r_{\rm E})-\mathbf{a}(\hat{\theta}_{\rm E},\hat{r}_{\rm E})\|_{2} \nonumber \\
		= &  \sqrt{2-2\mathcal{R}\{ \mathbf{a}^H(\theta_{{\rm E}},r_{\rm E}) \mathbf{a}(\hat{\theta}_{{\rm E}},\hat{r}_{\rm E}) \} }  
		\triangleq \varepsilon.\label{Exp:vecBound}
	\end{align}
	Then, we obtain the CSV errors in the angle/range domain, assuming no range/angle errors, respectively. 
	\begin{lemma}[CSV error in the range domain]\label{Lem:Rangebound}
		\rm Given the spherical uncertainty region $\mathcal{A}_{\rm E}$ and $\theta_{\rm E} = \hat{\theta}_{\rm E}$, the CSV error of Eve is given by
		\begin{align}\label{Exp:rangeBound}
			\varepsilon_{r} \triangleq \sqrt{2 - 2{C(\beta)}/{\beta}},
		\end{align}
		where $\beta = \sqrt{\frac{N^2d^2 (1-\sin^2\hat{\theta}_{\rm E}) }{2\lambda}\Big|\frac{1}{r_{\rm E}}-\frac{1}{\hat{r}_{\rm E}} \Big| }$ and $C(\beta)=\int_{0}^{\beta} \cos(\frac{\pi}{2}t^2) dt $.
		Moreover, 
		when $\Delta r_{\rm E}/\hat{r}_{\rm E}$ is small (e.g., $\Delta r_{\rm E}\le 0.1\hat{r}_{\rm E}$),  
		$\varepsilon_{r}$ in~\eqref{Exp:rangeBound} can be approximated as
		\begin{align}\label{Exp:linearR}
			\varepsilon_{r} \approx f_{r} (\hat{\theta}_{\rm E},\hat{r}_{\rm E}) \Delta r_{\rm E}  \triangleq \varepsilon_{r,{\rm appr}},
		\end{align}
		where $f_{r} (\hat{\theta}_{\rm E},\hat{r}_{\rm E}) =\frac{\pi N^2d^2 (1-\sin^2\hat{\theta}_{\rm E}) }{2\sqrt{20}\lambda \hat{r}_{\rm E}^2 }$.
	\end{lemma}
	\vspace{-5pt}
	\begin{proof}
		Please refer to Appendix~\ref{App:rangeBound}.
	\end{proof}
	
	\begin{lemma}[CSV error in the angle domain]\label{Lem:Anglebound}
		\rm Given the spherical uncertainty region $\mathcal{A}_{\rm E}$ and $ r_{\rm E} = \hat{r}_{\rm E}$, the CSV error of Eve is given~by 
		\begin{align}\label{Exp:angleBound}
			\varepsilon_{\theta} \triangleq \sqrt{2 - 2 \varpi},
		\end{align}
		where
		\begin{equation}
			\varpi  = \left\{
			\begin{aligned}
				&-\frac{1}{N \sin(\frac{3\pi}{2N})},\quad\textrm{if}~\sin\theta_{\rm E}\ge \sin\hat{\theta}_{\rm E} + \frac{3}{N}, \\
				&\frac{\sin\big(\frac{1}{2} N \pi (\sin\theta_{\rm E}-\sin\hat{\theta}_{\rm E}) \big)}{N \sin\big(\frac{1}{2}\pi (\sin\theta_{\rm E}-\sin\hat{\theta}_{\rm E}) \big)},~~ 	\textrm{otherwise}.
			\end{aligned} 
			\right.
		\end{equation}
		Moreover, when $ \sin\theta_{\rm E}-\sin\hat{\theta}_{\rm E} \le \frac{1}{2N}$,  $\varepsilon_{\theta}$ in~\eqref{Exp:angleBound} can be approximated~as
		\begin{align}\label{Exp:linearTheta}
			\varepsilon_{\theta} \approx f_{\theta}(\hat{\theta}_{\rm E})\Delta \theta_{\rm E} \triangleq \varepsilon_{\theta,{\rm appr}},
		\end{align}
		where $ f_{\theta}(\hat{\theta}_{\rm E})=   
		\frac{\pi N\cos \hat{\theta}_{{\rm E}} }{\sqrt{12}} $.
	\end{lemma}
	\vspace{-5pt}
	\begin{proof}
		Please refer to Appendix~\ref{App:angleBound}.
	\end{proof}
	
	The above two lemmas indicate that both $\varepsilon_{r}$ and $\varepsilon_{\theta}$ can be expressed as a \emph{linear} function w.r.t. $\Delta r_{\rm E}$ and $\Delta \theta_{\rm E}$,  when $ \Delta r_{\rm E} / \hat{r}_{\rm E}$ is small and $ \sin\theta_{\rm E} - \sin\hat{\theta}_{\rm E} \le \frac{1}{2N}$, respectively. 
	Based on the above, we obtain the conditions for which the first-order Taylor approximation of CSV is largely accurate.
	\begin{proposition}[Conditions for accurate first-order Taylor approximation]\label{Pro:JointTaylor}
		\rm 
		When $ \Delta r_{\rm E} / \hat{r}_{\rm E}$ is small (e.g., $\Delta r_{\rm E} / \hat{r}_{\rm E} \le 0.1$) and $ \sin\theta_{\rm E} - \sin\hat{\theta}_{\rm E} \le \frac{1}{2N}$, the first-order Taylor approximation
		$$\mathbf{a}(\theta_{\rm E},r_{\rm E})\approx \mathbf{a}(\hat{\theta}_{\rm E},\hat{r}_{\rm E}) +  \triangledown_{\theta}\mathbf{a}|_{(\hat{\theta}_{\rm E},\hat{r}_{\rm E})}\Delta\theta_{\rm E}
		+ \triangledown_{r}\mathbf{a}|_{(\hat{\theta}_{\rm E},\hat{r}_{\rm E})}\Delta r_{\rm E}$$
		is accurate.   Herein,  $\triangledown_{\theta}\mathbf{a}$ and $ \triangledown_{r}\mathbf{a} $ are the gradients of the CSV $\mathbf{a}(\theta,r)$ w.r.t. $\theta$ and ${r}$, which are given by
		\begin{subequations}\label{Exp:Gradient}
			\begin{align}
				\triangledown_{\theta}\mathbf{a} &= \jmath\cdot \mathbf{a} \odot [\frac{\partial \phi_1}{\partial \theta},\ldots,\frac{\partial \phi_N}{\partial \theta}]^T, \label{Exp:Taylorwrttheta}\\
				\triangledown_{r}\mathbf{a} &= \jmath\cdot \mathbf{a} \odot [\frac{\partial \phi_1}{\partial r},\ldots,\frac{\partial \phi_N}{\partial r}]^T, \label{Exp:Taylorwrtr}
			\end{align}
		\end{subequations}
		respectively, with $\phi_n \!=\! \frac{2\pi}{\lambda}(u_{n}\sin\theta \!-\! \frac{u_{n}^2\cos^2\theta}{2r})$, $\frac{\partial \phi_n}{\partial \theta} = \frac{2\pi}{\lambda} \left( u_n \cos \theta + \frac{u_n^2 \sin \theta \cos \theta}{r} \right)$, and  $ \frac{\partial \phi_n}{\partial r} =\frac{\pi u_n^2 \cos^2\theta}{\lambda r^2} $.
	\end{proposition}
	\begin{proof}
		Let   $ \|\Delta\mathbf{a}_{\rm E} \|_{2}^{2} $ denote the squared Euclidean norm of the approximation error, which is given by
		\begin{align}
			\|\Delta\mathbf{a}_{\rm E} \|_{2}^{2} &\triangleq \| \mathbf{a}(\theta_{\rm E},r_{\rm E}) -  \mathbf{a}(\hat{\theta}_{\rm E},\hat{r}_{\rm E})   \nonumber \\ &-\triangledown_{\theta}\mathbf{a}|_{(\hat{\theta}_{\rm E},\hat{r}_{\rm E})}\Delta\theta_{\rm E}
			- \triangledown_{r}\mathbf{a}|_{(\hat{\theta}_{\rm E},\hat{r}_{\rm E})}\Delta r_{\rm E}\|_{2}^{2}.  \nonumber
		\end{align}
		By introducing the  intermediate variable $\mathbf{a}(\hat{\theta}_{\rm E},{r}_{\rm E})$, $\|\Delta\mathbf{a}_{\rm E} \|_{2}^{2}$ can be rewritten as 
		\begin{align}
			\|\Delta\mathbf{a}_{\rm E} \|_{2}^{2} &=\| \mathbf{a}(\theta_{\rm E},r_{\rm E}) - \mathbf{a}(\hat{\theta}_{\rm E},r_{\rm E}) + \mathbf{a}(\hat{\theta}_{\rm E},r_{\rm E})
			-\mathbf{a}(\hat{\theta}_{\rm E},\hat{r}_{\rm E})  \nonumber\\
			&-\triangledown_{\theta}\mathbf{a}|_{(\hat{\theta}_{\rm E},\hat{r}_{\rm E})}\Delta\theta_{\rm E}
			- \triangledown_{r}\mathbf{a}|_{(\hat{\theta}_{\rm E},\hat{r}_{\rm E})}\Delta r_{\rm E}\|_{2}^{2}. 
		\end{align}
		Considering that $\triangledown_{\theta}\mathbf{a}|_{(\hat{\theta}_{\rm E},\hat{r}_{\rm E})}\Delta\theta_{\rm E} \approx \triangledown_{\theta}\mathbf{a}|_{(\hat{\theta}_{\rm E},{r}_{\rm E})}\Delta\theta_{\rm E} $ under a small range error, the proof of {\bf Proposition~\ref{Pro:JointTaylor}} can be equivalently established by proving 
		\begin{subequations}\label{Exp:partialTaylor}
			\begin{align}
				& \mathbf{a}({\theta}_{\rm E},\hat{r}_{\rm E}) \approx  \mathbf{a}(\hat{\theta}_{\rm E},\hat{r}_{\rm E}) + \triangledown_{\theta}\mathbf{a}|_{(\hat{\theta}_{\rm E},\hat{r}_{\rm E})}\Delta \theta_{\rm E}, \\
				& \mathbf{a}({\hat{\theta}}_{\rm E},r_{\rm E}) \approx \mathbf{a}(\hat{\theta}_{\rm E},\hat{r}_{\rm E}) + \triangledown_{r}\mathbf{a}|_{(\hat{\theta}_{\rm E},\hat{r}_{\rm E})} \Delta {r}_{\rm E}. 
			\end{align}
		\end{subequations}
		Let $\varepsilon_{\theta,{\rm Tayl}}  \triangleq \| \triangledown_{\theta}\mathbf{a}|_{(\hat{\theta}_{\rm E},\hat{r}_{\rm E})}\Delta \theta_{\rm E} \|_{2}  $, and $
		\varepsilon_{r,{\rm Tayl}}  \triangleq  \|\triangledown_{r}\mathbf{a}|_{(\hat{\theta}_{\rm E},\hat{r}_{\rm E})} \Delta {r}_{\rm E} \|_{2} $ denote the approximated CSV errors in the angle and range domains. When $ \Delta r_{\rm E} / \hat{r}_{\rm E}$ is small and $ \sin\theta_{\rm E} - \sin\hat{\theta}_{\rm E} \le \frac{1}{2N}$, we obtain that $\varepsilon_{\theta,{\rm Tayl}}  = \varepsilon_{\theta,{\rm appr}}$ and $\varepsilon_{r,{\rm Tayl}}  = \varepsilon_{r,{\rm appr}}$ (see Appendix~\ref{App:TaylorBound}). As such, the proof of~\eqref{Exp:partialTaylor} can be established via   Appendix~\ref{App:partialTaylor}. Based on the above,  we obtain $\|\Delta\mathbf{a}_{\rm E} \|_{2}^{2} \approx 0$, thus completing the proof.
	\end{proof}

	{\bf{Proposition 2}} establishes the conditions for which the first-order Taylor approximation of the near-field CSV under location uncertainty is largely accurate.
	Note that under the considered location error model,
	it can be easily shown that when the angle error satisfies $\sin(\hat{\theta}_{\rm E} +\Delta \theta_{\rm E})-\sin\hat{\theta}_{\rm E} \le \frac{1}{2N}$, the corresponding range error $ \Delta r_{\rm E} \ll \hat{r}_{\rm E}$ always holds, since ${\Delta \theta_{\rm E}} =\arcsin(\Delta r_{\rm E}/\hat{r}_{\rm E}) \approx \Delta r_{\rm E}/\hat{r}_{\rm E}$ is significantly small.
	The above findings indicate that to ensure the accuracy of the first-order Taylor approximation, one can partition the original uncertainty region $\mathcal{A}_{\rm E}$ into multiple sub-regions under the constraint that the  angle error satisfies 
	\begin{align}
		\sin(\hat{\theta}_{\rm E} +\Delta \theta_{\rm E})-\sin\hat{\theta}_{\rm E} \le \frac{1}{2N} \triangleq \mathcal{S}_{\theta},
	\end{align}
	which will be exploited in Section~\ref{Sec:III-B2} for the uncertain region partitioning.

	\begin{example}\rm 
		\begin{figure}[t]
			\centering
			\includegraphics[width=0.3\textwidth]{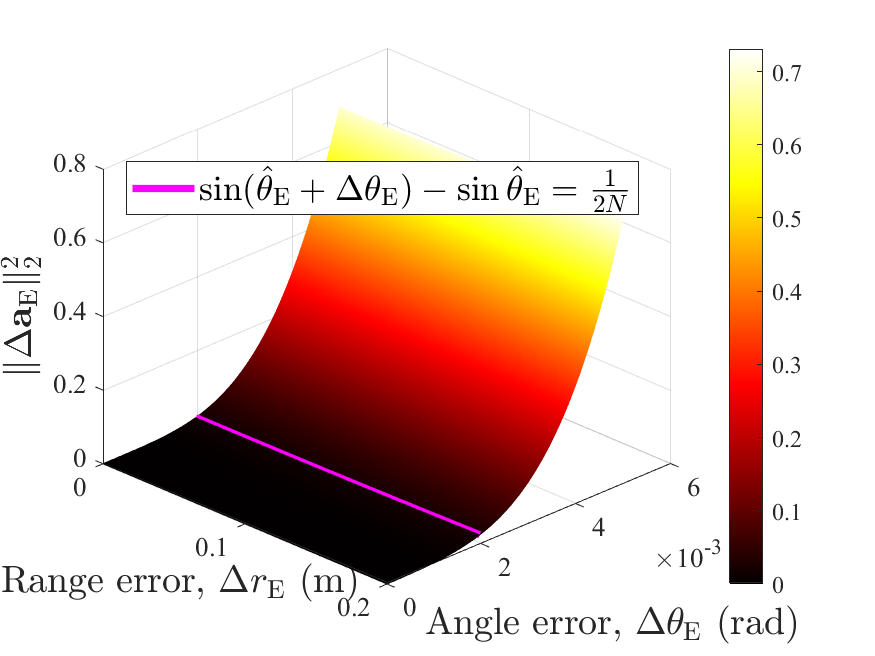}
			\caption{ $\|\Delta\mathbf{a}_{\rm E} \|_{2}^{2} $ versus angle error and range error.} \label{Fig:Sec3_accuracy}
			\vspace{-14pt}
		\end{figure}
		To evaluate the accuracy of the first-order Taylor approximation, we plot in Fig.~\ref{Fig:Sec3_accuracy} the squared Euclidean norm between the actual CSV and its Taylor-approximated vector, i.e.,$ \|\Delta\mathbf{a}_{\rm E} \|_{2}^{2}$. Key observations are made as follows.
		\begin{itemize}
			\item When  $\sin(\hat{\theta}_{\rm E} + \Delta \theta_{\rm E} ) -\sin\hat{\theta}_{\rm E} \le \frac{1}{2N}$,  $\|\Delta\mathbf{a}_{\rm E} \|_{2}^{2} $ remains relatively small, indicating that the Taylor approximation for the CSV is accurate in this regime.
			\item  Once $\sin(\hat{\theta}_{\rm E} + \Delta \theta_{\rm E}) - \sin \hat{\theta}_{\rm E} > \frac{1}{2N}$, the approximation error grows substantially, suggesting that the Taylor approximation becomes inaccurate and thus unsuitable for characterizing regions with a large angle error. 
		\end{itemize}
	\end{example}
		
	\subsubsection{Proposed Algorithm to Solve Problem (P2)}\label{Sec:III-B2}
	Based on the condition of accurate Taylor approximation, we propose an efficient two-stage robust beamforming method, which judiciously partitions the spherical uncertainty region in the first stage, followed by the second stage to reformulate and solve the resulting optimization problem by using the first-order Taylor approximation and the GSD techniques.

	\textbf{Stage 1 (Uncertainty region partitioning):}
	To ensure the accuracy of Taylor approximation in characterizing the CSV in the uncertainty location of Eve, we first divide the origin spherical uncertainty region $\mathcal{A}_{\rm E}$ into multiple \emph{fan-shaped} sub-regions, as shown in Fig.~\ref{Fig:Sampling}, which ensures that the maximum angle error within each sub-region is constrained by $	\sin(\hat{\theta}_{\rm E} +\Delta \theta)-\sin\hat{\theta}_{\rm E} \le \frac{1}{2N}$.
	Then, a surrogate location for each uncertainty sub-region is selected to facilitate robust beamforming design.
	
	\underline{{Angle region}}: First, given the estimated Eve location $(\hat{\theta}_{\rm E},\hat{r}_{\rm E})$ and its uncertainty region $\mathcal{A}_{\rm E}$, the largest and smallest angle for enabling accurate Taylor approximation can be determined as $\theta_{{\rm E}}^{(\rm ub)}= \hat{\theta}_{\rm E} + \arcsin(\Upsilon/\hat{r}_{\rm E})$ and $\theta_{{\rm E}}^{(\rm lb)} =\hat{\theta}_{\rm E} - \arcsin(\Upsilon/\hat{r}_{\rm E}$), respectively. Then, the angle uncertainty region $\big[\theta_{{\rm E}}^{(\rm lb)},\theta_{{\rm E}}^{(\rm ub)}\big]$ is partitioned into  $2S+1$ angle sub-regions, with their (central) sampling angles denoted as
	\begin{align}
		\tilde{\boldsymbol{\theta}}_{\rm E} \triangleq &\big[\varphi_{-S},\ldots,\varphi_{-1}, 
		\varphi_{0},
		\varphi_{1},\ldots,\varphi_{S} \big]^T,
	\end{align}
	Specifically, each angle sub-region is symmetric about its sampling angle $\varphi_{s}, \forall s\in \mathcal{S}$ with 
	$\varphi_{s}^{ \max}$ and $\varphi_{s}^{\min}$ being its maximum and minimum angles. To ensure the maximum angle error within each sub-region  constrained by $	\sin(\hat{\theta}_{\rm E} +\Delta \theta)-\sin\hat{\theta}_{\rm E} \le \frac{1}{2N} $, the angle sampling point and its corresponding maximum and minimum angles are set as
	\begin{subequations}
		\begin{align}
			\varphi_{s} &= \arcsin \Big( \sin\hat{\theta}_{{\rm E}}+\frac{s}{N} \Big), s \neq \pm \mathcal{S},\\
			\varphi_{s}^{ \max} &=  \arcsin \Big( \sin\hat{\theta}_{{\rm E}}+\frac{s}{N}+\frac{1}{2N} \Big), s\neq \pm \mathcal{S}, \\
			\varphi_{s}^{ \min} &=  \arcsin \Big( \sin\hat{\theta}_{{\rm E}}+\frac{s}{N}-\frac{1}{2N} \Big), s\neq \pm \mathcal{S}.
		\end{align}
	\end{subequations}	
	Additionally, $\varphi_{S} = (\theta_{{\rm E}}^{(\rm ub)} + \varphi_{S-1}^{ \max})/2$ with $(\varphi_{S}^{ \min},\varphi_{S}^{ \max}) = (\varphi_{S-1}^{ \max},\theta_{{\rm E}}^{(\rm ub)})$, $\varphi_{-S} = (\theta_{{\rm E}}^{(\rm lb)} + \varphi_{-S+1}^{ \min})/2$ with $(\varphi_{-S}^{ \min},\varphi_{-S}^{ \max}) = (\theta_{{\rm E}}^{(\rm lb)}, \varphi_{-S+1}^{ \min})$, and $S = \lfloor (\sin\hat{\theta}_{{\rm E}}^{(\rm ub)} -  \sin \hat{\theta}_{\rm E} )/2\mathcal{S}_{\theta} + 1/2 \rfloor$.

	\underline{Range region}: Next, given each angle uncertainty region $\mathcal{A}^{(\theta)}_{s} \triangleq [\varphi_{s}^{ \min},\varphi_{s}^{ \max}], s \in \mathcal{S}$, the corresponding range uncertainty region can be obtained as $\mathcal{A}^{(r)}_{s} \triangleq [r_{s}^{ \min},r_{s}^{ \max}], s \in \mathcal{S}$. 
	Herein,  $r_{s}^{\max}$ and $r_{s}^{\min}$ denote the maximum and minimum ranges of Eve within the overall uncertainty region $\mathcal{A}_{\rm E}$ and angular sub-region $[\varphi_{s}^{ \min},\varphi_{s}^{ \max}]$, which are given by
	\begin{subequations}
		\begin{align}
			r_{s}^{\rm \min} = \min\{ \bar{d}_{{\rm E},1}(\varphi_{s}^{\min}),\bar{d}_{{\rm E},1}(\varphi_{s}^{\max
			}) \},\\
			r_{s}^{\rm \max} = \max\{ \bar{d}_{{\rm E},2}(\varphi_{s}^{\min}),\bar{d}_{{\rm E},2}(\varphi_{s}^{\max
			}) \},
		\end{align}
	\end{subequations}
respectively. Herein, $\bar{d}_{{\rm E},1}(\theta)=\hat{d}_{\rm E}(\theta) -\sqrt{\Upsilon^2-\big(\check{d}_{\rm E}(\theta)\big)^{2}}$ and $\bar{d}_{{\rm E},2}(\theta)=\hat{d}_{\rm E}(\theta) +\sqrt{\Upsilon^2-\big(\check{d}_{\rm E}(\theta)\big)^{2}}$ with $\hat{d}_{\rm E}(\theta) = \hat{x}_{\rm E}\cos\theta + \hat{y}_{\rm E}\sin\theta$ and $\check{d}_{\rm E}(\theta) = \hat{x}_{\rm E}\sin\theta - \hat{y}_{\rm E}\cos\theta$, respectively. We thus obtain the fan-shaped uncertainty sub-regions  $\mathcal{A}_{s}, s \in \mathcal{S}$, as shown in Fig.~\ref{Fig:Sampling}, which is mathematically expressed as
	\begin{align}
		\mathcal{A}_{s} =\big\{ (\theta,r)| \theta \in [\varphi_{s}^{\min},\varphi_{s}^{\max}],~r\in[r_{s}^{\rm \min},r_{s}^{\rm \max}] \big\}.
	\end{align}
	
	\begin{figure}[t]
		\centering
		\includegraphics[width=0.3\textwidth]{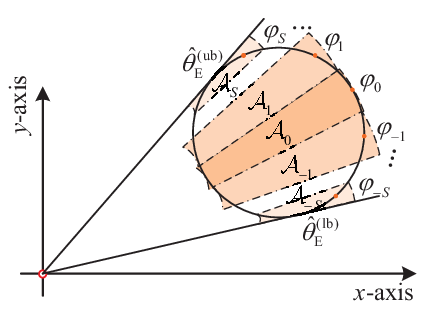}
		\caption{Schematic of the sub-region partitioning.} \label{Fig:Sampling}
		\vspace{-12pt}
	\end{figure}
	\underline{Surrogate location selection}: 
	Last, for each sub-region $\mathcal{A}_s$, the sampling angle $\varphi_s$ is retained as the surrogate angle. Additionally, the midpoint of the range uncertainty region is set as the surrogate range, i.e., $ r_{s}  = ({r_{s}^{\min}  + r_{s}^{\max}})/{2},  s \in \mathcal{S}$. Based on the above, the surrogate location of each sub-region $\mathcal{A}_{s}$ is set as $(\varphi_{s}, r_{s}),  s \in \mathcal{S}$. As such, for each sub-region, given the surrogate location $(\varphi_{s}, r_{s})$, the maximum absolute angle and range errors within $ \mathcal{A}_{s}$ satisfy 
	\begin{subequations}
		\begin{align}
			&\big|{\Delta r_{s}}\big| \le \left| ({r_{s}^{\max}  - r_{s}^{\min}})/{2}  \right| \triangleq \epsilon_{s},~s\in\mathcal{S}, \\
			&|\Delta \varphi_{s}| \le  \left|(\varphi_{s}^{\max} -  \varphi_{s}^{\min})/2 \right| \triangleq \vartheta_{s},~s\in\mathcal{S}.  
		\end{align}
	\end{subequations} 
	
	\textbf{Stage 2 (Refined LMI-based robust beamforming design):}
	Based on the uncertainty region partitioning, the optimization problem \textbf{(P2)} is reformulated as
	\begin{align}
		(\textbf{P3}):\; \max_{\mathbf{w}} \quad & g\big(\mathbf{w}|\mathbf{w}^{(j)}\big) \nonumber \\
		\text{s.t.}\quad     &   \eqref{C:Power_S}, \nonumber \\
		& \max_{ (\theta_{\rm E}, r_{\rm E}) \in \mathcal{A}_{s} }  |\mathbf{a}^H(\theta_{\rm E}, r_{\rm E}) \mathbf{w} |^2 \!\le\! \Gamma,~ \forall  s\in \mathcal{S}.  \label{C:Secrecy_S_Trans}
	\end{align}
	Note that although constraint~\eqref{C:Secrecy_S} is reformulated as $2S+1$ continuous uncertainty sets, Problem \textbf{(P3)} remains intractable due to its infinite number of constraints in~\eqref{C:Secrecy_S_Trans}. Based on the analysis in Section~\ref{Sec:III-B1}, after sub-region partitioning,  the CSV of Eve $(\theta_{\rm E},r_{\rm E})\in \mathcal{A}_{s}$ can be approximated by using the first-order Taylor expansion around $\mathbf{a}(\varphi_{s}, r_{s}), s\in \mathcal{S}$. Moreover, to address the issue of conservative beamforming design in the error-bound-based method, for which near-field angle-error amplification results in a significant error bound, we propose a refined LMI method as follows to achieve a much higher secure transmission rate.
	\begin{proposition}[Refined LMI reformulation]\label{Pro:LMI}
		\rm
		Based on the Taylor approximation of CSV in~\eqref{Exp:NFV_Taylor},  constraint~\eqref{C:Secrecy_S_Trans} can be rewritten as
		\begin{align}\label{Exp:SU_set}
			\max_{\boldsymbol{\zeta}_{s}^T \boldsymbol{\Sigma}_{s}^{-1} \boldsymbol{\zeta}_{s} \le 1,~\forall s\in \mathcal{S}}  \quad  \big|\mathbf{w}^H \big(\mathbf{a}(\varphi_{s},r_{s}) + \mathbf{J}_{s} \boldsymbol{\zeta}_{s}      \big) \big|^2 \le \Gamma,  
		\end{align}
		where $\mathbf{J}_{s} = [\triangledown_{r}\mathbf{a},\triangledown_{\theta}\mathbf{a}]\big|_{( \varphi_{s}, r_{s} )}$ is the gradient matrix,
		$ \boldsymbol{\zeta}_{s} = [\Delta r_{s},  \Delta \varphi_{s}]^T $ is the error vector,  
		 and $\boldsymbol{\Sigma}_{s}=\diag(\epsilon_{s}^2,\vartheta_{s}^2) $. 
	Then, by denoting $\mathbf{a}_{s} \triangleq \mathbf{a}(\varphi_{s},r_{s}) $ for simplicity, the non-convex 
	constraint~\eqref{Exp:SU_set} can be re-expressed as the following refined $2S+1$ LMIs by using GSD
	\begin{align}\label{Exp:LMI}
		\mathbf{H}_{s} \!\triangleq \!\left[
		\begin{array}{@{}c@{~}c@{~}c@{}}
			\left[
			\begin{array}{c@{~~}c@{}}
				{\Gamma} & \mathbf{w}^{H}\mathbf{a}_{s} \\
				\mathbf{a}_{s}^{H}\mathbf{w} & 1-{\lambda}_{s}^{(r)}-{\lambda}_{s}^{(\theta)} 
			\end{array}
			\right] 
			& \epsilon_{s} {\mathbf{b}}_{s}^{(r)} & \vartheta_{s} {\mathbf{b}}_{s}^{(\theta)} \\[1pt]
			\epsilon_{s} ({\mathbf{b}}_{s}^{(r)})^H & {\lambda}_{s}^{(r)}  & {0} \\[1pt]
			\vartheta_{s} ({\mathbf{b}}_{s}^{(\theta)})^H & {0}  & {\lambda}_{s}^{(\theta)} 
		\end{array}
		\right] \! \!\succeq\! \mathbf{0},
	\end{align}
	where ${\lambda}_{s}^{(r)}\ge0$ and $ {\lambda}_{s}^{(\theta)}  \ge 0$ are auxiliary variables,  ${\mathbf{b}}_{s}^{(r)} = [\triangledown_{r}\mathbf{a}^H\big|_{( \varphi_{s}, r_{s} )}\mathbf{w}, 0 ]^H$, and  ${\mathbf{b}}_{s}^{(\theta)} = [\triangledown_{\theta}\mathbf{a}^H\big|_{( \varphi_{s}, r_{s} )}\mathbf{w}, 0 ]^H$.
	\end{proposition}
	\vspace{-5pt}
	\begin{proof}
	Please refer to Appendix~\ref{App:LMI}.
	\end{proof}
	
	Based on {\bf{Proposition 3}}, the non-convex constraint~\eqref{C:Secrecy_S_Trans} involving infinite uncertainty sets is transformed into a finite set of deterministic LMIs. As such, Problem \textbf{(P3)} can be recast as the following convex problem
	\begin{align}
		(\textbf{P4}):\; \max_{\mathbf{w},{\boldsymbol{\lambda}}^{(r)},\boldsymbol{\lambda}^{(\theta)}}  \quad & g\big(\mathbf{w}|\mathbf{w}^{(j)}\big) \nonumber \\
		\text{s.t.} \quad   &  \eqref{C:Power_S}, \nonumber \\
		&\mathbf{H}_{s} \succeq \mathbf{0},~\forall  s\in \mathcal{S},
	\end{align}
	where ${\boldsymbol{\lambda}}^{(r)} = [{\lambda}_{1}^{(r)},\ldots,{\lambda}_{S}^{(r)}]^T$ and ${\boldsymbol{\lambda}}^{(\theta)}  = [{\lambda}_{1}^{(\theta)} ,\ldots,{\lambda}_{S}^{(\theta)} ]^T$. \textbf{(P4)} is a convex optimization problem, which can be efficiently solved by using convex optimization tools, e.g., CVX.

	\begin{figure*}[t]
		\centering
		\begin{subfigure}[b]{0.245\linewidth}
			\includegraphics[width=1.05\linewidth]{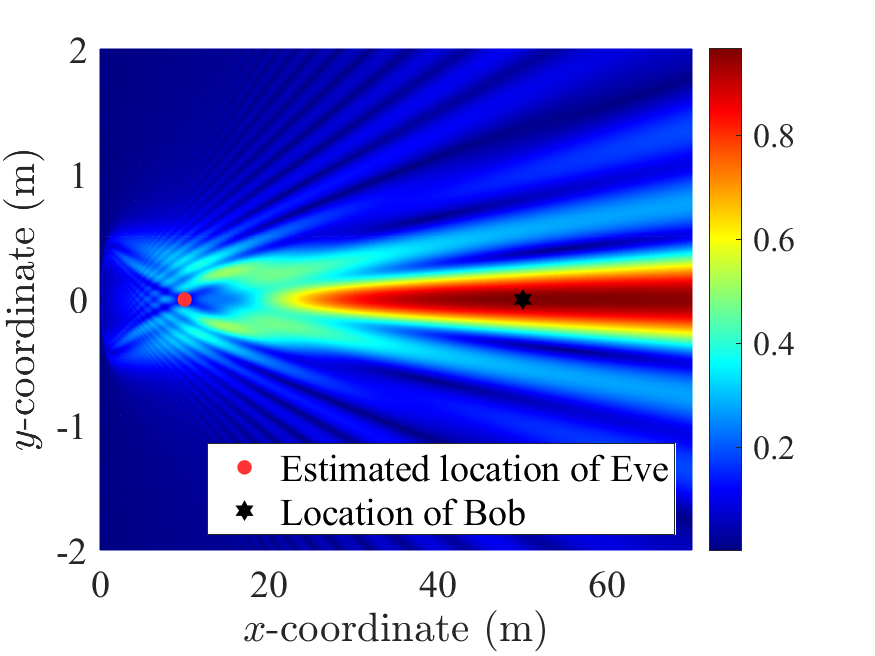}
			\caption{Non-robust method.}
			\label{Fig:SUMapping1}
		\end{subfigure}
		\hspace{-5pt}
		\begin{subfigure}[b]{0.245\linewidth}
			\includegraphics[width=1.05\linewidth]{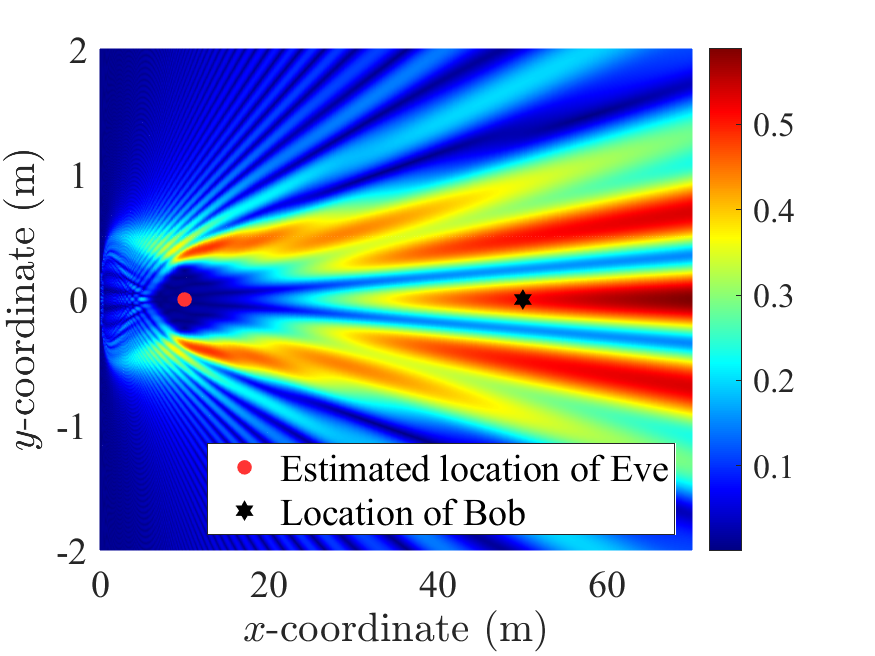}
			\caption{Sampling-based method.}
			\label{Fig:SUMapping2}
		\end{subfigure}
		\hspace{-5pt}
		\begin{subfigure}[b]{0.245\linewidth}
			\includegraphics[width=1.05\linewidth]{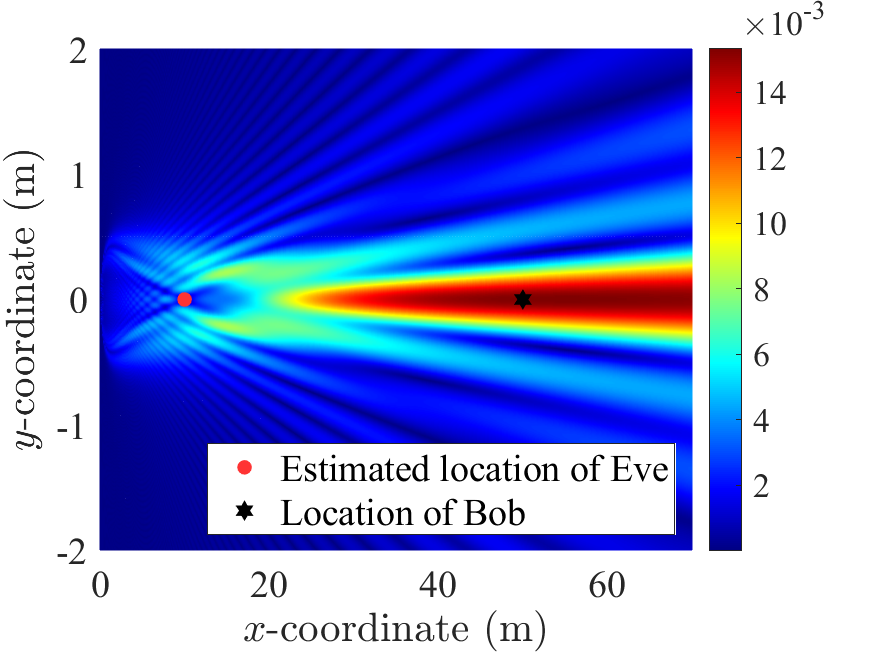}
			\caption{Error-bound-based method.}
			\label{Fig:SUMapping3}
		\end{subfigure}
		\hspace{-5pt}
		\begin{subfigure}[b]{0.245\linewidth}
			\includegraphics[width=1.05\linewidth]{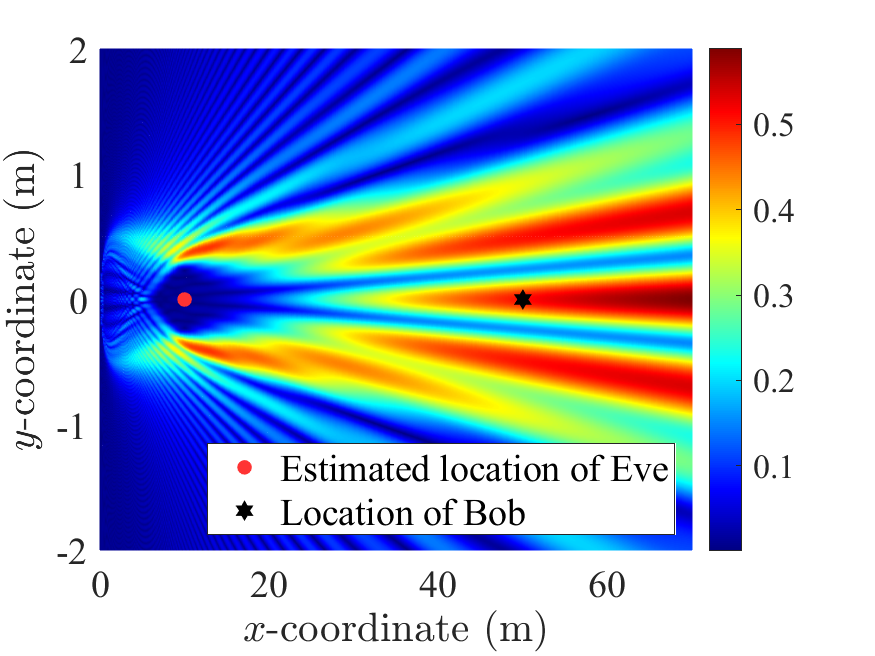}
			\caption{Proposed method.}
			\label{Fig:SUMapping4}
		\end{subfigure}
		\caption{Beam patterns of different methods.}
		\label{Fig:SUMapping}
		\vspace{-12pt}
	\end{figure*}
	
		\begin{remark}[Advantages of  proposed method]\rm 
		Compared with the conventional methods in Section~\ref{Sec:III-A}, the proposed near-field robust beamforming design enjoys two key advantages. First, unlike existing methods~\cite{li2023robust} where the location uncertainty is characterized by its CSV error bound  (e.g., $ \varepsilon_{{\rm Tayl}}^{(\rm ub)} $ or $ \varepsilon^{(\rm ub)} $), which is significantly large when the range of Eve is relatively small,
		the proposed method incorporates both the surrogate CSV $\mathbf{a}(\varphi_{s}, r_{s})$ and its associated gradient space $\mathbf{J}_{s} \boldsymbol{\zeta}_{s}$ in~\eqref{Exp:SU_set}, which together characterize the location uncertainty region of the Eve. As such, the proposed method is capable of designing efficient beamforming vector to reduce the information leakage within the uncertain region of Eve, rather than reducing the BS transmit power,  hence achieving an increased rate.
		Second, according to the GSD lemma~\cite{GuiZhou_Robust}, the dimension of the resulting LMI is determined by the dimension of the underlying uncertainty variables. In the error-bound-based method, the  uncertainty vector is of dimension $N$, leading to an LMI involving a matrix $\mathbf{G}_{\rm E}$ of size $N+2$ in~\eqref{Exp:Con_LMI}. In contrast, the proposed refined LMI method characterizes the location uncertainty (i.e., $\mathcal{A}_{\rm E}$) using only two variables (range and angle). As a result, the dimension of the matrix $\mathbf{H}_{\rm s}$ in~\eqref{Exp:LMI} reduces from $(N+2)$ to $4$ (see {\bf Proposition~\ref{Pro:LMI}} for details), thereby significantly reducing the computational complexity. 
		
	\end{remark} 
 
	\begin{figure}[t]
		\centering
		\begin{subfigure}{0.49\linewidth} 
			\centering
			\includegraphics[width=\linewidth]{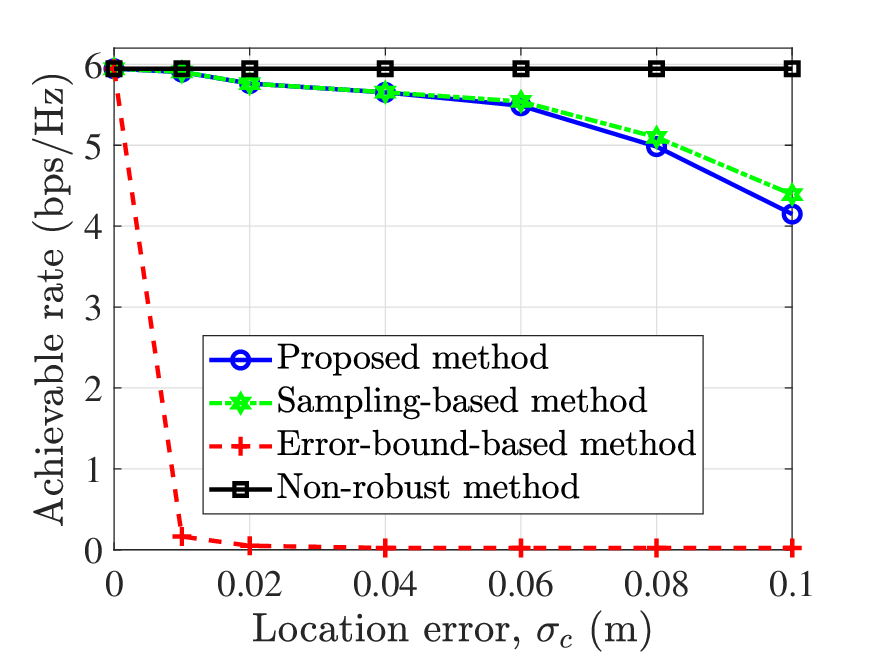}
			\caption{Achievable rate versus location error.}
			\label{Fig:Sec3-Rate}
		\end{subfigure}
		\hspace{-5pt} 
		\begin{subfigure}{0.49\linewidth} 
			\centering
			\includegraphics[width=\linewidth]{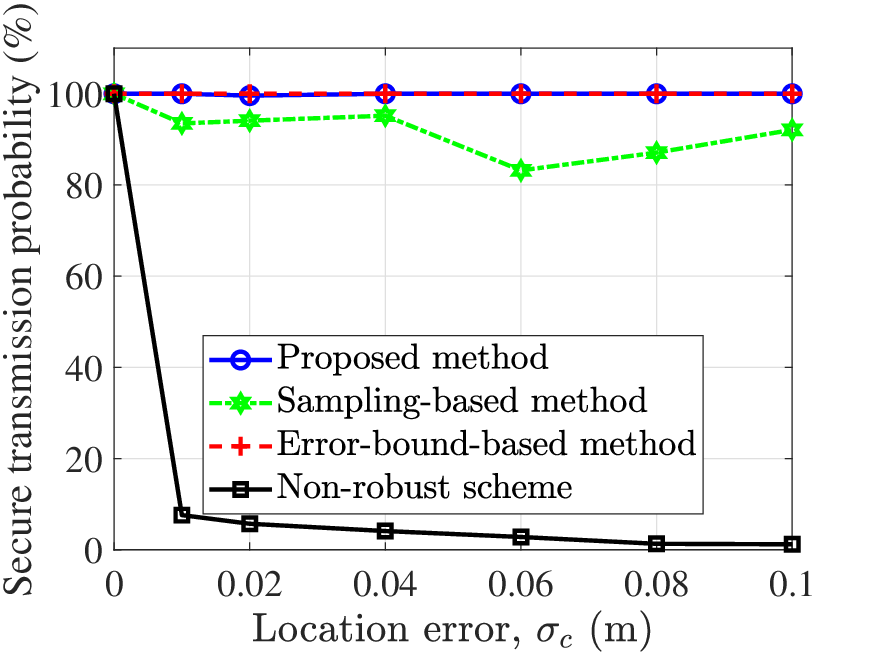}
			\caption{Secure transmission probability versus location error.}
			\label{Fig:Sec3-Feas}
		\end{subfigure}
		\caption{Achievable rate and secure transmission probability versus location error.}
		\label{Fig:Sec3_RateandFeas}
		\vspace{-15pt} 
	\end{figure}
	
	\begin{example}\rm 
		In Figs.~\ref{Fig:SUMapping} and~\ref{Fig:Sec3_RateandFeas}, we present the beam patterns, achievable rate, and secure transmission probability of different methods where $\sigma_{{\rm c}}= 0.1$, $\hat{\theta}_{\rm E} = 0$, $\hat{r}_{\rm E} = 10$ m, ${\theta}_{\rm B} = 0$, and  ${r}_{\rm B} = 50$ m. Key observations are summarized below:
		\begin{itemize}
			\item 
		\emph{Non-robust method}: 
       This method achieves the highest rate shown in Fig.~\ref{Fig:Sec3_RateandFeas}(a), but it suffers from a low secure transmission probability (see Fig.~\ref{Fig:Sec3_RateandFeas}(b)) because of severe power leakage around the estimated location, as illustrated in Fig.~\ref{Fig:SUMapping}(a).
	
		\item 
		\emph{Sampling-based method}: 
		This method attains the second-highest achievable rate with a secure transmission probability above 80\% (see Fig.~\ref{Fig:Sec3_RateandFeas}). As shown in Fig.~\ref{Fig:SUMapping}(b), it markedly reduces the power leakage around the estimated location at the huge computational cost.
        	
		\item 
		\emph{Error-bound-based method}: This method
        ensures a secure transmission probability of 1 by conservatively constraining the BS transmit power. This yields a beam pattern similar to that of the non-robust method, while the beamforming gain of  former one is a much smaller than the latter one, as shown in Fig.~\ref{Fig:SUMapping}(c). Consequently, it achieves the lowest achievable rate (see Fig.~\ref{Fig:Sec3_RateandFeas}(a)).

		\item 
		\emph{Proposed method}: 
        The proposed method attains close achievable rate with the sampling-based method, while ensuring a secure transmission probability of 1, as shown in Fig.~\ref{Fig:Sec3_RateandFeas}.
        By partitioning the uncertainty region and applying a refined LMI-based beamforming design, it suppresses power leakage in the potential region of Eve as shown in Fig.~\ref{Fig:SUMapping}(d).
		\end{itemize}
	\end{example}

	\section{Multiple-Bob-Multiple-Eve Scenario}
	In this section, we extend the robust beamforming design in Section~\ref{Sec:III} to the more general scenario with multiple Bobs and multiple Eves. Specifically, in the first stage, the spherical uncertainty regions of all Eves are partitioned. Then, in the second stage, a refined LMI reformulation is performed, upon which the SCA technique is applied to solve the resulting optimization problem.

	\vspace{-6pt}
	\subsection{Proposed Solution to Problem \textbf{(P1)}}
	Note that compared with the single-user case in Section~\ref{Sec:III}, the near-field robust beamforming design in the general case is more challenging, since 1) the worst-case eavesdropping rate of each Bob at each Eve must be guaranteed, which imposes stricter constraints on the system design, and 2) the existence of multi-user interference complicates the objective function, making the robust optimization problem highly challenging to solve.
	
	To address these issues,  we first reformulate Problem \textbf{(P1)} as follows
	\begin{align}
		(\textbf{P5}):\; \max_{\{\mathbf{w}_{k}\}} \quad &\sum_{k=1}^{K}~ R_{{\rm B},k}  \nonumber \\
		\text{s.t.}   \quad &\eqref{C:Power},\nonumber \\
		&\max_{\{\mathbf{q}_{{\rm E},m}\in \mathcal{A}_{m}\}}  |\mathbf{a}(\mathbf{q}_{{\rm E},m})\mathbf{w}_{k}|^2 \!\le\! \Gamma_{m},~\forall m, k,   \label{C:Secrecy_MU}
	\end{align}
	where $\Gamma_{m}=\frac{\sigma^2 (2^{R_{\max}}-1)}{N|h_{{\rm E},m}|^2}$.  To solve this non-convex problem, the SCA technique is employed to obtain a concave surrogate function of ${R}_{{\rm B},k}$ in the objective function, which is given by~\cite{MyTPM}
	\begin{align}
		\tilde{R}_{{\rm B},k} & \triangleq \frac{1}{\ln2} \Bigg(  
		\ln  \bigg( 1+\frac{| {\iota}_{k}^{\left(j \right) } |^2}{  {\nu}_{k}^{\left(j \right) }  } \bigg)     
		- \frac{| {\iota}_{k}^{\left(j \right) } |^2}{  {\nu}_{k}^{\left(j \right) }  } 
		\nonumber \\
		& + \frac{2\mathcal{R}\left\lbrace (  {\iota}_{k}^{\left(j \right) } )^H  {\iota}_{k}  \right\rbrace }{{\nu}_{k}^{\left(j \right) }}
		-\frac{| {\iota}_{k}^{\left(j \right) } |^2 \big(  | {\iota}_{k} |^2  + {\nu}_{k}  \big)  }
		{ {\nu}_{k}^{\left(j \right) } \big(| {\iota}_{k}^{\left(j \right) } |^2+   {\nu}_{k}^{\left(j\right) }\big) }
		\Bigg),
	\end{align}
	with ${\iota}_{k} = \mathbf{h}^H_{{\rm B},k}  \mathbf{w}_{k}$,  ${\nu}_{k} = \sum_{i=1,i\neq k}^{K} |\mathbf{h}^H_{{\rm B},k}  \mathbf{w}_{i} |^2 + \sigma^2$, $ {\iota}_{k}^{(j)} =  \mathbf{h}^H_{{\rm B},k}  \mathbf{w}_{k}^{(j)} $, and ${\nu}_{k}^{(j)} = \sum_{i=1,i\neq k}^{K} |\mathbf{h}^H_{{\rm B},k}  \mathbf{w}_{i}^{(j)} |^2 + \sigma^2 $. Herein, $\mathbf{w}_{k}^{(j)}$ denotes the beamforming vector obtained from the $j$-th SCA iteration. Based on the above, we construct a tight concave lower bound for ${R}_{{\rm B},k}$ (i.e., ${R}_{{\rm B},k} \ge \tilde{R}_{{\rm B},k}$).
	
	Next, to address the non-convex constraint~\eqref{C:Secrecy_MU} induced by continuous uncertainty sets, we extend the method devised in Section~\ref{Sec:III-B2} to this more general scenario. Specifically, for an arbitrary Eve $m$, the spherical uncertainty region $\mathcal{A}_{{\rm E},m}$ is partitioned into multiple sub-region  $\mathcal{A}_{m,s_{m}}$, with their surrogate locations being
	\begin{subequations}
		\begin{align}
			\tilde{\boldsymbol{\theta}}_{{\rm E},m} \triangleq &\big[\varphi_{m,-S_m},\ldots, 
			\varphi_{m,0},
			\ldots,\varphi_{m,S_m} \big]^T,~\forall m, \\
			\tilde{\mathbf{r}}_{{\rm E},m} \triangleq &\big[r_{m,-S_m},\ldots,
			r_{m,0},
			\ldots,r_{m,S_m} \big]^T,~\forall m.
		\end{align}
	\end{subequations}
	Herein, $s_{m}\in\mathcal{S}_{m}$, where $\mathcal{S}_{m}\triangleq \{-S_{m},\ldots,0,\ldots,S_{m}\} $ with $2S_{m}+1$ being  the number of  angle sampling  points for Eve~$m$.  
	Based on {\bf Proposition~\ref{Pro:LMI}},   constraint~\eqref{C:Secrecy_MU} can be re-expressed as follows.
	\begin{corollary}
		\rm 
		By performing a Taylor expansion of the CSV at each surrogate location, constraint~\eqref{C:Secrecy_MU} can be re-expressed~as
		\begin{align}\label{Exp:MU_set}
			\max_{\boldsymbol{\zeta}_{m,s_{m}}^T \!\!\boldsymbol{\Sigma}_{m,s_{m}}^{-1}\! \boldsymbol{\zeta}_{m,s_{m}} \!\le 1}  	&\big|\mathbf{w}_{k}^H \big(\mathbf{a}(\varphi_{m,s_{m}},r_{m,s_{m}}) \!+\! \mathbf{J}_{m,s_{m}} \boldsymbol{\zeta}_{m,s_{m}}      \big) \big|^2  \nonumber \\& \le \Gamma_{m},~\forall m\in \mathcal{M},k\in \mathcal{K},s_{m}\in \mathcal{S}_{m}, 
		\end{align}
		where $ \boldsymbol{\zeta}_{m,s_{m}} \!=\! [\Delta r_{m,s_{m}}, \Delta \varphi_{m,s_{m}}]^T $,   
		$\mathbf{J}_{m,s_{m}} \!=\! [\triangledown_{r} \mathbf{a},\triangledown_{\theta}\mathbf{a}]\big|_{( \varphi_{m,s_{m}}, r_{m,s_{m}} )}$,  
		and $\boldsymbol{\Sigma}_{m,s_{m}}\!=\!\diag(\epsilon_{m,s_{m}}^2,\vartheta_{m,s_{m}}^2) $, the values of which can be determined based on the method proposed in Section~\ref{Sec:III-B2}. 
	\end{corollary}
	
	Based on {\bf{Proposition}~\ref{Pro:LMI}}, the above constraint~\eqref{Exp:MU_set} can be transformed as the following refined LMIs
	\begin{align}\label{C:MULMIs}
		\mathbf{H}_{m,k,s_{m}} \succeq \mathbf{0},~\forall m \in \mathcal{M}, k\in \mathcal{K}, s_{m}\in \mathcal{S}_{m},
	\end{align}
	where $\mathbf{H}_{m,k,s_{m}}$ is given by~\eqref{Exp:Gen_LMI} at the top of next page.
	\setcounter{equation}{\value{equation}} 
	\begin{figure*}[ht]
		\begin{equation}\label{Exp:Gen_LMI}
			\begin{aligned}
				\mathbf{H}_{m,k,s_{m}} \triangleq \left[
				\begin{array}{@{}c@{~}c@{~}c@{}}
					\left[
					\begin{array}{c@{~~}c@{}}
						{\Gamma}_{m}  & \mathbf{w}_{k}^{H}\mathbf{a}_{m,s_{m}} \\
						\mathbf{a}_{m,s_{m}}^{H}\mathbf{w}_{k} & 1-{\lambda}_{m,k,s_{m}}^{(r)}-{\lambda}_{m,k,s_{m}}^{(\theta)} 
					\end{array}
					\right] 
					& \epsilon_{m,s_{m}} {\mathbf{b}}_{m,k,s_{m}}^{(r)} & \vartheta_{m,s_{m}} {\mathbf{b}}_{m,k,s_{m}}^{(\theta)} \\[1pt]
					\epsilon_{m,s_{m}} ({\mathbf{b}}_{m,k,s_{m}}^{(r)})^H & {\lambda}_{m,k,s_{m}}^{(r)}  & 0 \\[1pt]
					\vartheta_{m,s_{m}} ({\mathbf{b}}_{m,k,s_{m}}^{(\theta)})^H & 0 & {\lambda}_{m,k,s_{m}}^{(\theta)} 
				\end{array}
				\right].
			\end{aligned}
		\end{equation}
		\hrulefill 
		\vspace{-18pt} 
	\end{figure*}  
	Herein, ${\lambda}_{m,k,s_{m}}^{(r)}\ge0$ and $ {\lambda}_{m,k,s_{m}}^{(\theta)}  \ge 0$ are auxiliary variables, ${\mathbf{b}}_{m,k,s_{m}}^{(r)} = [\triangledown_{r}\mathbf{a}^H\big|_{( \varphi_{m,s_{m}}, r_{m,s_m} )}\mathbf{w}_{k}, 0 ]^H$, and  ${\mathbf{b}}_{m,k,s_{m}}^{(\theta)} = [\triangledown_{\theta}\mathbf{a}^H\big|_{( \varphi_{m,s_{m}}, r_{m,s_m} )}\mathbf{w}_{k}, 0 ]^H$.
	
	As such, Problem \textbf{(P5)} can be reformulated as follows
	\begin{align}
		(\textbf{P6}):\; \max_{\{\mathbf{w}_{k}\},\{{\boldsymbol{\Lambda}}_{m}^{(r)}\},\{\boldsymbol{\Lambda}_{m}^{(\theta)}\}} \quad &\sum_{k=1}^{K}~ \tilde{R}_{{\rm B},k}  \nonumber \\
		\text{s.t.} \quad  \quad \quad  \quad &\eqref{C:Power},~\eqref{C:MULMIs}, \nonumber
	\end{align} 
	where ${\boldsymbol{\Lambda}}_{m}^{(r)} \in \mathbb{R}^{K\times S_{m}}$ and $\boldsymbol{\Lambda}_{m}^{(\theta)} \in \mathbb{R}^{K\times S_{m}}$ are auxiliary  matrices with $[{\boldsymbol{\Lambda}}_{m}^{(r)}]_{k,s_{m}} = {\lambda}_{m,k,s_{m}}^{(r)} \ge 0$ and $[{\boldsymbol{\Lambda}}_{m}^{(\theta)} ]_{k,s_{m}} = {\lambda}_{m,k,s_{m}}^{(\theta)}  \ge 0$, respectively. 
	Now, Problem (\textbf{P6}) becomes a convex optimization problem, which can be efficiently solved via the CVX tool to obtain an optimal solution. 
	
	\begin{remark}[Algorithm convergence and computational complexity] \rm First, we analyze the convergence performance of the proposed algorithm.  Since the achievable sum-rate obtained at each SCA iteration is non-decreasing, i.e., $\tilde{R}_{{\rm B},k}^{(j)} \ge \tilde{R}_{{\rm B},k}^{(j-1)}$, the convergence of the proposed algorithm is guaranteed. Second, we evaluate the computational complexity. In each SCA iteration for solving Problem \textbf{(P6)}, the proposed algorithm involves $v = KM \sum_{m=1}^{M} (2S_m + 1)$ LMI constraints of size 4, and one second-order cone (SOC) constraint of size $NK$. Since $\sum_{m=1}^{M} (2S_m + 1)$ is typically much smaller than $N$, the dominant number of variables is in the order of $\mathcal{O}(\eta)$ with $\eta = NK$. According to~\cite{MyTPM}, the computational complexity for solving Problem \textbf{(P6)} is in the order of $\mathcal{O}\big(J\sqrt{4v+2}\cdot \log(1/\bar{\epsilon})(4^3\eta v + 4^2 \eta^2  v + \eta(KN)^2+ \eta^3)\big)$, where  $\bar{\epsilon}$ is the predefined precision of CVX solver and $J$ represents the number of SCA iterations.
	\end{remark}

	\begin{remark}[Multi-path scenarios]\label{Dis:Multipath}
	\rm 
	For multi-path scenarios, we assume that the perfect CSIs of Bobs are available at the BS by using existing channel estimation methods~\cite{Cui2022CE}. However, the CSI of Eve $m$, given by
	\begin{align}
		\mathbf{h}_{{\rm E},m}^{H}
		=\sqrt{N}h_{{\rm E},m} \mathbf{a}^{H}(\mathbf{q}_{{\rm E},m}) +\mathbf{h}_{{\rm E,NLoS},m}^{H},
	\end{align}
	is generally difficult to acquire.  Given that only estimated locations of Eves are available and the power of NLoS paths is typically much weaker than that of the LoS path,  we assume that the NLoS component is bounded by the following condition: $\| \mathbf{h}_{{\rm E,NLoS},m}\|_{2} \le \varepsilon_{{\rm E,NLoS},m} = \kappa_{{\rm E},m} \|\sqrt{N}h_{{\rm E},m} \mathbf{a}^{H}(\mathbf{q}_{{\rm E},m})\|_2$,
	where $\varepsilon_{{\rm E,NLoS},m}$  represents the maximum possible power for the NLoS paths and $\kappa_{{\rm E},m} \in (0,1)$ denotes the power ratio between the NLoS and LoS paths.  As such, a conservative secrecy constraint can be obtained as follows.
	\begin{align}
		&\quad \max_{ \{\mathbf{q}_{{\rm E},m} \in \mathcal{A}_{m} \}}  ~|\mathbf{h}_{{\rm E},m}^H \mathbf{w}_{k} |^2  \nonumber \\
		&\le \max_{ \{\mathbf{q}_{{\rm E},m} \in \mathcal{A}_{m} \}}  \!\! \big\{N |h_{{\rm E},m} |^2 |\mathbf{a}^H({\mathbf{q}}_{{\rm E},m}) \mathbf{w}_{k} |^2 \big\}\!+ \! \varepsilon_{{\rm E,NLoS},m}^{2}\mathbf{w}_{k}^H\mathbf{w}_{k} \nonumber \\
		&\le \sigma^2 (2^{R_{\max}}-1),~\forall m\in \mathcal{M}, k\in \mathcal{K}.
	\end{align}
	Based on the above, the proposed scheme can be employed for robust beamforming design in multi-path scenarios.	
\end{remark}

\begin{remark}[Location uncertainty of Bobs]\label{Dis:BobLoc_error}\rm
    For~scenarios accounting for the location uncertainty of Bobs,  robust beamforming design becomes more challenging due to the complicated objective function induced by such uncertainty. We consider the worst-case sum-rate maximization problem, for which its  optimization objective is given by $\max_{\{\mathbf{w}_{k}\}}~\sum_{k=1}^{K}~\big\{\min_{\{\mathbf{q}_{{\rm B},k}\in \mathcal{A}_{{\rm B},k}\}} {R_{{\rm B},k}} \big\}$, which aims to maximize the sum of the minimum achievable rates under the location uncertainty of Bobs. To tackle this challenge, we first reformulate it into an equivalent form as follows.
		\begin{subequations}
			\begin{align}
				\max_{\{\mathbf{w}_{k},\tau_{k,1},\tau_{k,2}\}}\quad & \sum_{k=1}^{K}\log_{2}\Big(1 + \frac{\tau_{k,1}}{\tau_{k,2}+\sigma^2}\Big) \nonumber \\
				\text{s.t.}\quad\quad\quad & \min_{\{\mathbf{q}_{{\rm B},k}\in \mathcal{A}_{{\rm B},k}\}} | \mathbf{h}_{{\rm B},k}^{H} \mathbf{w}_{k} |^2  \ge \tau_{k,1}, \forall k, \nonumber \\
				& \max_{\{\mathbf{q}_{{\rm B},k}\in \mathcal{A}_{{\rm B},k}\}} \sum_{i=1,i\neq k}^{K} | \mathbf{h}_{{\rm B},k}^{H} \mathbf{w}_{i}   |^{2} \le \tau_{k,2}, \forall k, \nonumber
			\end{align}
		\end{subequations}
		where $\tau_{k,1}$ and $\tau_{k,2}$ are auxiliary variables. Although this is a more complex problem, the proposed two-stage robust beamforming method can be extended to this scenario by
		using the SCA and S-Procedure techniques~\cite{GuiZhou_Robust}.
		The detailed derivations, however, are mathematically intricate and thus are left for future work.
\end{remark}

	\vspace{-6pt}
	\section{Numerical Results}\label{Sec:SR}
	In this section, we present numerical results to demonstrate the efficiency of the proposed robust beamforming scheme in ensuring near-field secure transmission.
	
	\vspace{-6pt}
	\subsection{System Setup and Benchmark Schemes}
	We consider an XL-array system at  $f = 30 $ GHz frequency band, where the BS equipped with $N=256$ antennas serves $K=2$ Bobs in the presence of $M = 2$ Eves. This configuration aligns with typical setups adopted in existing literature~\cite{Cui2022CE,liu2025physical}. The estimated Eve locations are $(\hat{x}_{{\rm E},1},\hat{y}_{{\rm E},1}) = (10, 0.5)$ m and $(\hat{x}_{{\rm E},2},\hat{y}_{{\rm E},2}) = (10, -0.5)$ m with identical location errors  $\sigma_{c,1}  =\sigma_{c,2}=\sigma_{c}=0.1$ m. The Bobs are randomly distributed within a circular region of radius $\Upsilon_{\rm {B}} = 3$ m, which is centered at $(50, 0)$ m. Unless otherwise specified, the system parameters are set as $P_{\max} = 1$ W, $R_{\max} = 1$ bps/Hz, $\sigma^2 = -60$ dBm, and $\alpha = 0.05$.
    \begin{figure}[t]
		\centering
		\includegraphics[width=0.3\textwidth]{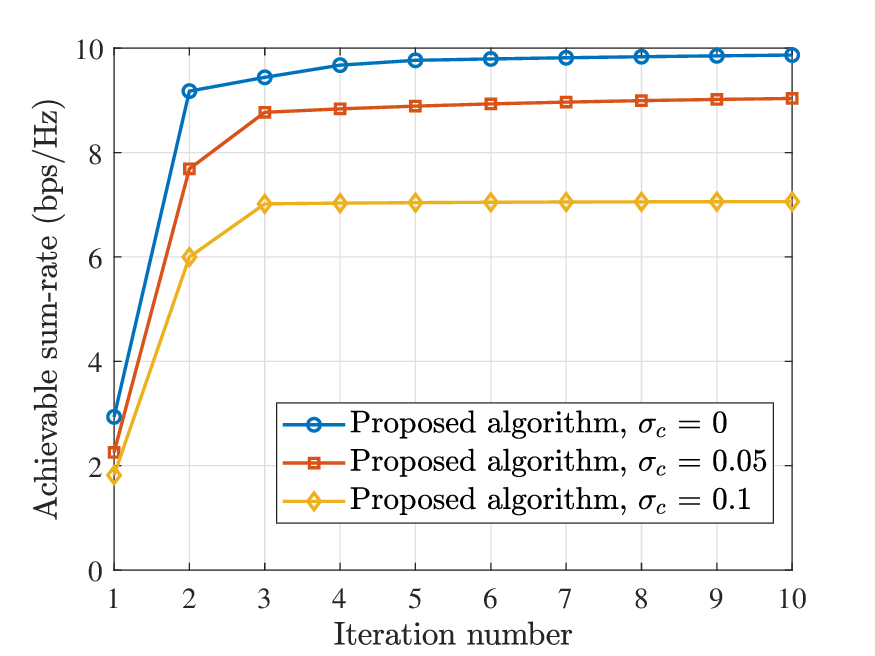}
		\caption{Achievable sum-rate versus iteration number.} \label{Fig:Sec5Convergence}
		\vspace{-12pt}
	\end{figure}
    \begin{figure}[t]
    	\centering
    	\includegraphics[width=0.32\textwidth]{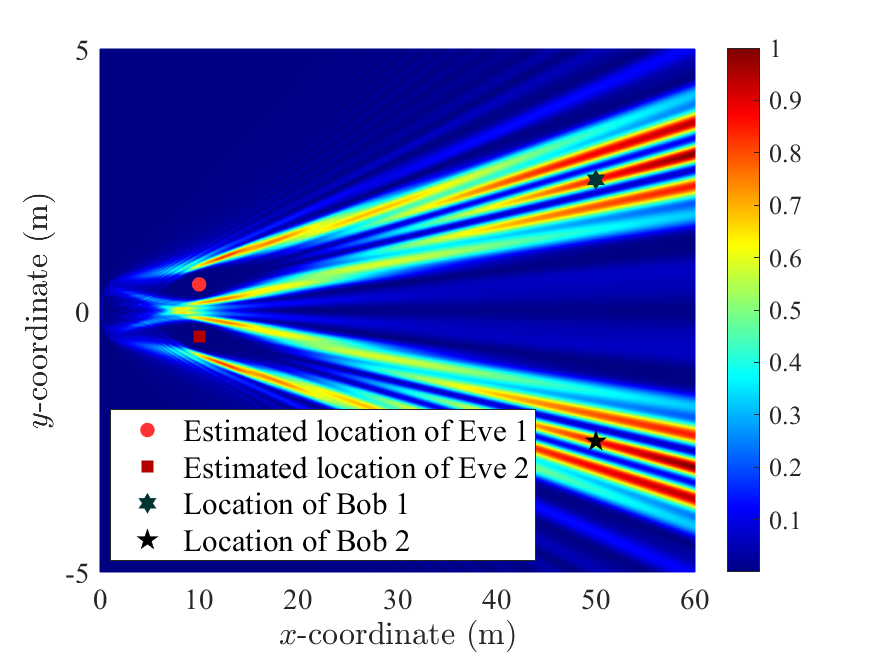}
    	\caption{Beam pattern of the proposed scheme.} \label{Fig:Sec5Mupattern1}
    	\vspace{-16pt}
    \end{figure}
	For performance comparison, the following benchmark schemes are considered.
	\begin{itemize}
		\item \emph{Non-robust scheme}: This scheme designs the beamforming vectors based on the estimated location of Eve solely, without considering the location uncertainty issue.
		
		\item \emph{Sampling-based scheme}: In this scheme, the angle uncertainty region $[\theta_{\rm E,m}^{(\rm lb)},\theta_{\rm E,m}^{(\rm ub)}] $ is uniformly sampled with $\bar{S}_{m}=100$, based on which the beamforming vectors are designed.
		
		\item \emph{Error-bound-based scheme}: In this scheme, the beamforming vectors are designed based on GSD lemma with the accurate error bound $\varepsilon_{m}^{(\rm {ub})}$.
		
		\item \emph{Uncertainty region partitioning only scheme}: For this scheme, the uncertainty region is partitioned into multiple fan-shaped sub-regions. The CSV error bound for each sub-region (i.e., $\varepsilon_{m,s_m}^{(\rm ub)}$) is then used for beamforming design via a conventional LMI reformulation.
		
		\item \emph{Refined LMI only scheme}: This scheme employs the refined LMI reformulation for robust beamforming design, yet without partitioning the uncertainty region.
		
	\end{itemize} 

	\vspace{-6pt}
	\subsection{Performance Analysis}
	\subsubsection{Convergence of the proposed algorithm} To demonstrate the convergence of the proposed algorithm, Fig.~\ref{Fig:Sec5Convergence} plots the achievable sum-rate against the iteration number under different location errors $\sigma_{\rm c}$. The results show that the algorithm converges to a stationary point within nearly 10 iterations for all considered values of $\sigma_{\rm c}$.

	\subsubsection{Beam pattern of proposed scheme} 
	In Fig.~\ref{Fig:Sec5Mupattern1}, we present the beam pattern of the proposed scheme for the two Bobs located at $(x_{{\rm B},1},y_{{\rm B},1})=(50, 2.5)$ m and $(x_{{\rm B},2},y_{{\rm B},2})=(50, -2.5)$ m. It is observed that both Bobs achieve satisfactory beam gain, even when one is Eve positioned at the same spatial angle. This is attributed to the near-field effect, which allows for flexible beamforming design across joint angle-range domains. Importantly, the beam gains at the estimated Eve locations and their surrounding areas are significantly suppressed, as our proposed robust beamforming design effectively minimizes energy leakage within the uncertainty region for each Eve to meet the worst-case secrecy constraint.

	\subsubsection{Effect of location error} 
	\begin{figure}[t]
		\centering
		\begin{subfigure}{0.4\textwidth}
			\centering
			\includegraphics[width=0.9\linewidth]{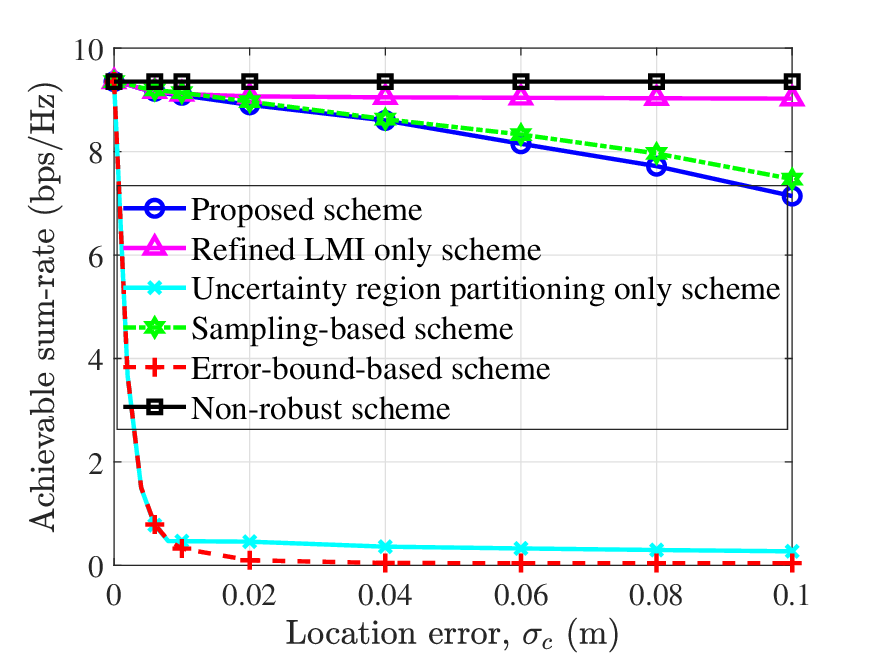}
			\caption{Achievable sum-rate versus location error.}
			\label{Fig:Sec5Rate_Error}
		\end{subfigure}
		\hspace{15pt} 
		\begin{subfigure}{0.4\textwidth}
			\centering
			\includegraphics[width=0.9\linewidth]{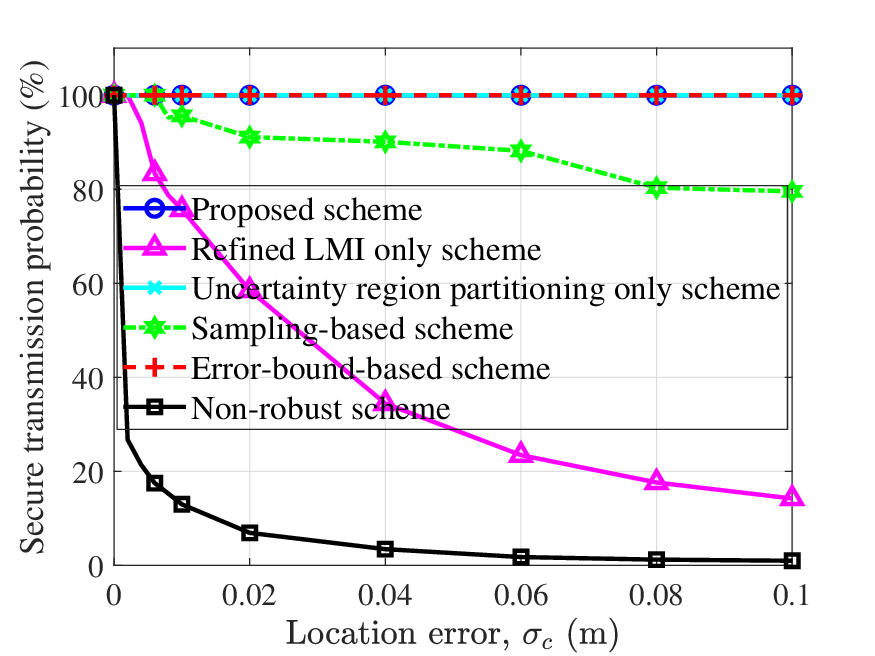}
			\caption{Secure transmission probability versus location error.}
			\label{Fig:Sec5Feas_Error}
		\end{subfigure}
		\caption{Achievable rate and secure transmission probability versus power ratio.}
		\label{Fig:Error}
		\vspace{-15pt} 
	\end{figure}
	
	To show the effect of location error, we plot in Figs.~\ref{Fig:Error}(a) and~\ref{Fig:Error}(b) the curves of achievable sum-rate and secure transmission probability versus the location error $\sigma_{{\rm c}}$ for all schemes, respectively. The non-robust scheme achieves the highest sum-rate by enforcing secrecy constraints only at the estimated locations of the Eves. This allows to increase the achievable sum-rate, while it results in a significantly low secure transmission probability, as shown in Fig.~\ref{Fig:Error}(b). On the other hand, both the error-bound-based scheme and the uncertainty region partitioning only scheme suffer from significantly degraded rate performance as $\sigma_{\rm c}$ increases, which approaches zero when $\sigma_{\rm c} > 0.02$. This is because the CSV error bound in these two schemes remains large even for small location errors (e.g., $\sigma_{\rm c} = 0.02$), thereby imposing strict constraints on the optimized BS transmit power.
	The proposed scheme, along with the refined LMI-only and sampling-based schemes, demonstrates similar achievable rates, as shown in Fig.~\ref{Fig:Error}(a). However, for the latter two schemes, the secure transmission probability is significantly lower as $\sigma_{\rm c}$ increases. These results indicate that both the uncertainty region partitioning and the refined LMI reformulation contribute to improved secrecy performance. Among these, although the achievable sum-rate of the proposed scheme gradually decreases with $\sigma_{\rm c}$, it strikes a favorable trade-off between rate performance and secrecy under location uncertainty, significantly outperforming the other benchmark schemes in terms of both robustness and rate performance.

	\subsubsection{Effect of power for NLoS paths} 
	\begin{figure}[t]
		\centering
		\begin{subfigure}{0.4\textwidth}
			\centering
			\includegraphics[width=0.9\linewidth]{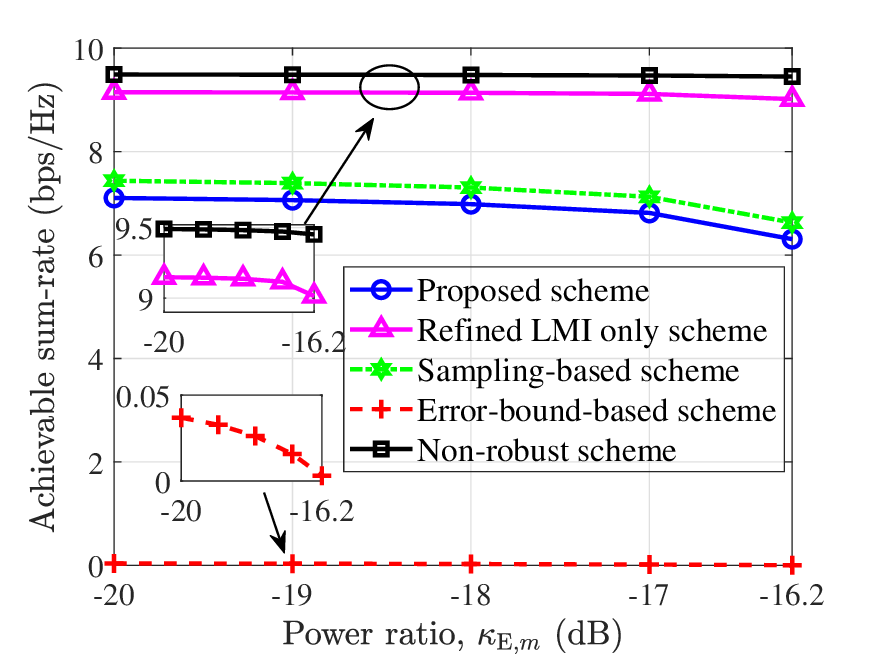}
			\caption{Achievable sum-rate versus power ratio.}
			\label{Fig:Sec5Rate_Powratio}
		\end{subfigure}
		\hspace{15pt} 
		\begin{subfigure}{0.4\textwidth}
			\centering
			\includegraphics[width=0.9\linewidth]{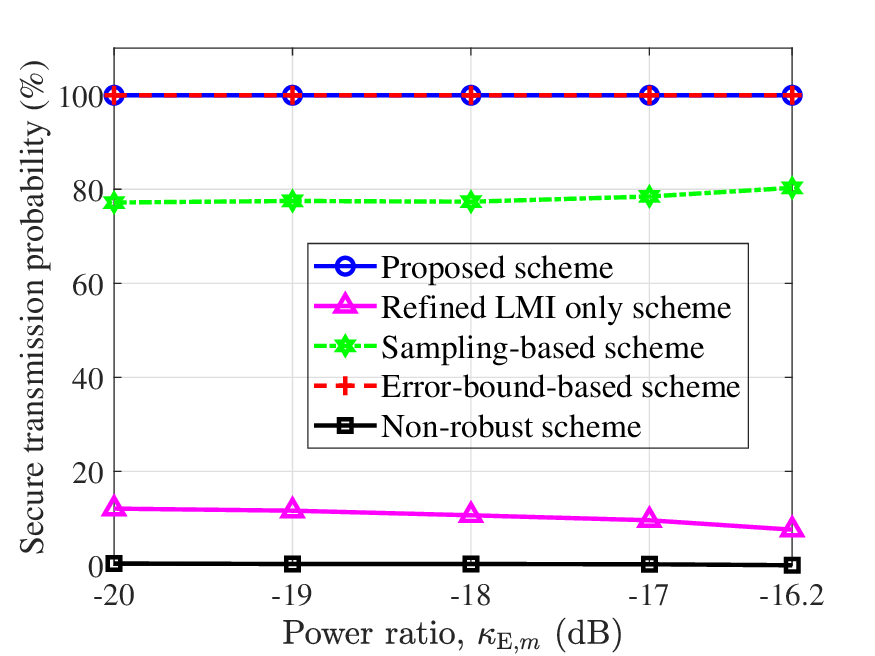}
			\caption{Secure transmission probability versus power ratio.}
			\label{Fig:Sec5Feas_Powratio}
		\end{subfigure}
		\caption{Achievable rate and secure transmission probability versus power ratio.}
		\label{Fig:Powratio}
		\vspace{-15pt} 
	\end{figure}
	
	Last, we show in Figs.~\ref{Fig:Powratio}(a) and~\ref{Fig:Powratio}(b) the performance of different schemes versus the maximum possible power for NLoS paths as discussed in {\bf Remark~\ref{Dis:Multipath}}. 
	Similar observations are made: the non-robust scheme achieves the highest achievable sum-rate, while it suffers from a very low secure transmission probability. In contrast, our proposed scheme ensures a secure transmission probability of $1$, while achieving  favorable rate performance. 
	Moreover, the achievable sum-rate of all five schemes decreases with the power ratio $\kappa_{{\rm E},m}$.  
	These results suggest that in multi-path scenarios, stronger NLoS components impose more conservative secrecy constraints (see {\bf Remark~\ref{Dis:Multipath}}), thereby degrading the achievable sum-rate, particularly for the sampling-based and proposed schemes.

	\vspace{-6pt}
	\section{Conclusions}\label{Sec:Con}
        In this paper, we studied \emph{robust} beamforming design for near-field PLS systems, under the assumption of perfect location information of Bobs and imperfect location information of Eves. To solve the formulated sum-rate maximization problem under worst-case eavesdropping rate constraints, we first revealed the \emph{near-field angular-error amplification} effect, which significantly degrades the performance of conventional robust beamforming schemes. To tackle this challenge, we analyzed the accuracy of first-order Taylor approximation for the near-field CSV and established its region of sufficient accuracy. Based on these insights, we proposed a two-stage robust beamforming framework: the first stage partitions the uncertainty region, and the second stage employs a refined LMI-based reformulation. This method was further extended to the multi-Bob multi-Eve scenario. Last, numerical results validated the robustness and superiority of the proposed scheme in near-field PLS systems. The proposed robust beamforming design can be extended to scenarios such as extremely large-scale intelligent reflecting surface (XL-IRS)-aided systems,  space-air-ground integrated networks, and beam tracking for ISAC.
        
	\begin{appendices}  
		\section{}\label{App:PowerUB}
		A necessary condition for the feasibility of the optimization problem corresponding to the error-bound-based method is $\mathbf{G}_{\rm E} \succeq \mathbf{0}$.
		To ensure $\mathbf{G}_{\rm E} \succeq \mathbf{0}$, the following inequalities should be guaranteed
		{\small\begin{align}
		\left|
			\begin{array}{@{}c@{}c@{}}
				{\Gamma} - \hat{\lambda}_{\rm E} &  \mathbf{a}(\hat{\theta}_{\rm E},\hat{r}_{\rm E})^H\mathbf{w}   \\
				\mathbf{w}^H\mathbf{a}(\hat{\theta}_{\rm E},\hat{r}_{\rm E})  &     1    
			\end{array}
			\right| \ge 0,~  
			\left|\begin{array}{@{}c@{~~}c@{}}
				1     & \varepsilon_{\rm Tayl}^{(\rm ub)} \mathbf{w}^H \\
				\varepsilon_{\rm Tayl}^{(\rm ub)} \mathbf{w}  & \hat{\lambda}_{\rm E} \mathbf{I}_{(N\times N)} 
			\end{array}\right| \ge 0.\nonumber 
		\end{align}}Specifically, the first determinant inequality guarantees the positive semi-definiteness of the top-left submatrix of $\mathbf{G}_{\rm E}$, while the second ensures that of the bottom-right submatrix.
		The first determinant inequality can be rewritten as $	\hat{\lambda}_{\rm E} \le \Gamma$. Moreover, 
		based on $\det\Big(\begin{array}{@{}c@{~}c@{}} 
			\mathbf{A} &\mathbf{B}\\
			\mathbf{C} &\mathbf{D}\\
		\end{array}
		\Big) = \det(	\mathbf{A} ) \det(	\mathbf{D} -\mathbf{C} 	\mathbf{A}^{-1}   \mathbf{B} )  $,  the second determinant inequality can be rewritten as $(\hat{\lambda}_{\rm E})^{N-1}\Big(\hat{\lambda}_{\rm E} - \big(\varepsilon_{\rm Tayl}^{(\rm ub)}\big)^2 \|\mathbf{w}\|_2^2 \Big) \ge 0$. For any $\Gamma \ge \hat{\lambda}_{\rm E}>0$, we have $\|\mathbf{w}\|_{2}^{2} \le {\Gamma}/{\big(\varepsilon_{\rm Tayl}^{(\rm ub)}\big)^2}$, thus completing the proof.

		\section{}\label{App:rangeBound}
		According to~\cite{Cui2022CE}, when $\theta_{\rm E} = \hat{\theta}_{\rm E}$, the term $\mathbf{a}^H(\hat{\theta}_{{\rm E}},r_{\rm E}) \mathbf{a}(\hat{\theta}_{{\rm E}},\hat{r}_{\rm E})$ can be approximated as $ \frac{C(\beta)+\jmath S(\beta)}{\beta} $ with 
		$\beta = \sqrt{\frac{N^2d^2 (1-\sin^2\hat{\theta}_{\rm E}) }{2\lambda}\Big|\frac{1}{r_{\rm E}}-\frac{1}{\hat{r}_{\rm E}} \Big| }$, 
		$C(\beta)=\int_{0}^{\beta} \cos(\frac{\pi}{2}t^2) dt $, 
		and  $S(\beta)=\int_{0}^{\beta}\sin(\frac{\pi}{2}t^2)$. 
		Based on the above,  the term $\mathcal{R}\{\mathbf{a}^H(\hat{\theta}_{{\rm E}},r_{\rm E}) \mathbf{a}(\hat{\theta}_{{\rm E}},\hat{r}_{\rm E})\} $ can be rewritten as ${C(\beta)}/{\beta}$, leading to $\varepsilon_{r}$ in~\eqref{Exp:rangeBound}.
		
		When $\Delta r_{\rm E} /\hat{r}_{\rm E}$  is sufficiently small to ensure that $\beta^2 \le 1/\pi$, we have
		$\cos (\frac{\pi}{2}t^2)\approx 1 - \frac{\pi^2}{8}t^4$ for $0\le t\le \beta $, and ${C(\beta)}/{\beta}$ in ~\eqref{Exp:rangeBound} can be approximated as
		$	\frac{C(\beta)}{\beta} \approx \frac{\int_{0}^{\beta} \big(1 -\frac{\pi^2}{8}t^4  \big) dt }{\beta} = 1 - \frac{\pi^2}{40}\beta^4$. 
		Consequently, $\varepsilon_{r}$ can be rewritten as
		\begin{align}\label{Eq:varepsilon_R}
			\varepsilon_{r} &\approx  \frac{\pi N^2d^2 (1-\sin^2\hat{\theta}_{\rm E}) }{2\sqrt{20}\lambda} \Big|\frac{1}{{r}_{\rm E} }-\frac{1}{\hat{r}_{\rm E}} \Big|. 
		\end{align}
		By substituting $r_{\rm E} = \hat{r}_{\rm E}-\Delta r_{\rm E}$ into~\eqref{Eq:varepsilon_R}, we obtain  
		\begin{align}
			\varepsilon_{r} \overset{(b)}{\approx} f_{r}(\hat{\theta}_{\rm E},\hat{r}_{\rm E}) \Delta r_{\rm E}, 
		\end{align}
		where $ f_{r}(\hat{\theta}_{\rm E},\hat{r}_{\rm E}) = \frac{\pi N^2d^2 (1-\sin^2\hat{\theta}_{\rm E}) }{2\sqrt{20}\lambda \hat{r}_{\rm E}^2 } $, and $(b)$ holds by using 
		$\frac{1}{\hat{r}_{\rm E} - \Delta r_{\rm E}} \approx \frac{1}{\hat{r}_{\rm E}} + \frac{\Delta r_{\rm E}}{\hat{r}_{\rm E}^2}$ under the condition  $ \Delta r_{\rm E} \ll \hat{r}_{\rm E}$ (e.g., $\Delta r_{\rm E} \le 0.1 \hat{r}_{\rm E}$), thus completing the proof.
		
		\vspace{-6pt}
		\section{}\label{App:angleBound}
		When $r_{\rm E} = \hat{r}_{\rm E}$, the approximation  
		$	\frac{1-\sin^2\theta_{{\rm E}}}{r_{\rm E}} -\frac{1-\sin^2\hat{\theta}_{{\rm E}}}{\hat{r}_{\rm E}} \approx  0 $
		holds due to the relatively small angle error. As such, the term $\mathbf{a}^H({\theta}_{{\rm E}},\hat{r}_{\rm E}) \mathbf{a}(\hat{\theta}_{{\rm E}},\hat{r}_{\rm E})$ can be approximated as
		\begin{align}
			&\mathbf{a}^H({\theta}_{{\rm E}},\hat{r}_{\rm E}) \mathbf{a}(\hat{\theta}_{{\rm E}},\hat{r}_{\rm E})
			=\frac{1}{N} \sum_{n=1}^{N} e^{\jmath\frac{2\pi}{\lambda}\big(u_{n} (\sin\theta_{{\rm E}}- \sin \hat{\theta}_{{\rm E}}  ) \big)} \nonumber \\
			=&  \frac{\sin\big(\frac{1}{2} N \pi (\sin\theta_{\rm E}-\sin\hat{\theta}_{\rm E}) \big)}{N \sin\big(\frac{1}{2}\pi (\sin\theta_{\rm E}-\sin\hat{\theta}_{\rm E}) \big)} \triangleq \Xi(\sin\theta_{\rm E}-\sin\hat{\theta}_{\rm E}).
		\end{align}
		Given that $\Xi(\sin\theta_{\rm E}-\sin\hat{\theta}_{\rm E})$ decreases monotonically over the interval between $\sin\hat{\theta}_{\rm E}$ and $\sin\hat{\theta}_{\rm E} + \frac{3}{N}$, and reaches its minimum around $\sin\hat{\theta}_{\rm E} + \frac{3}{N}$, we thus obtain  $\varpi$ in~\eqref{Exp:angleBound}.
		
		By denoting $ \xi  = \sin\theta_{\rm E}-\sin\hat{\theta}_{\rm E}$,  the term $\varpi$ in~\eqref{Exp:angleBound} is given by
		$	\varpi = \frac{\sin(\frac{1}{2}N\pi\xi)}{N\sin(\frac{1}{2}\pi\xi)}. $
		Given $\xi \leq \frac{1}{2N}$ and $N \gg 1$, it follows that $0 \leq \frac{1}{2}N\pi\xi \leq \frac{\pi}{4}$ and $\frac{1}{2}\pi\xi \approx 0$. Using the approximations $\sin x \approx x - \frac{1}{6}x^3$ for $0 < x < \pi/4$ and $\sin y \approx y$ for $y \approx 0$,  $\varpi $ can be approximated~as
		$	\varpi \approx \frac{N \frac{1}{2}\pi \xi}{N\frac{1}{2}\pi \xi}\Big(1-\frac{1}{6} \big(\frac{1}{2}N\pi \xi\big)^2\Big) = 1-\frac{\pi^2N^2 \xi^2}{24}$.
		Consequently,  $\varepsilon_{\theta}$ can be rewritten as
		\begin{align}
			\varepsilon_{\theta}^{(\rm ub)} 
			&  = \frac{\pi N}{\sqrt{12}} \big(\sin (\hat{\theta}_{\rm E} + \Delta \theta_{\rm E} ) - \sin\hat{\theta}_{\rm E} \big) \nonumber \\
			& \overset{(c)}{\approx} \frac{\pi N}{\sqrt{12}} 
			\cos \hat{\theta}_{\rm E}  \Delta \theta_{\rm E}  \triangleq f_{\theta}(\hat{\theta}_{\rm E}) \Delta \theta_{\rm E},
		\end{align}
		where $f_{\theta}(\hat{\theta}_{\rm E}) =\frac{\pi N}{\sqrt{12}} \cos \hat{\theta}_{\rm E} $, and
		$(c)$ holds due to  $\cos \Delta \theta_{\rm E} \approx 1$ and $\sin \Delta \theta_{\rm E} \approx \Delta \theta_{\rm E}$ for small $\Delta \theta$,  thus completing the proof.
		
		\vspace{-6pt}
		\section{}\label{App:TaylorBound}
		The CSV error in the range domain obtained via Taylor approximation (i.e., $\varepsilon_{r,{\rm Tayl}}$)  can be expressed as
		\begin{align}
			&\varepsilon_{r,{\rm Tayl}} = \| \triangledown_{r}\mathbf{a}|_{(\hat{\theta}_{\rm E},\hat{r}_{\rm E})}\Delta r_{\rm E} \|_{2}  \nonumber \\
			=&\sqrt{\sum_{n=1}^N \frac{1}{N} \Big( \frac{\pi u_n^2 \cos^2 \hat{\theta}_{\rm E}}{\lambda \hat{r}_{\rm E}^2}  \Big)^2} \Delta r_{\rm E} =  \frac{\pi\cos^2\hat{\theta}_{\rm E}}{\lambda \hat{r}_{\rm E}^2 \sqrt{N}}\sqrt{\sum_{n=1}^{N} u_n^4 } \Delta r_{\rm E}. \nonumber
		\end{align}
		Based on antenna positions $u_{n},n\in\mathcal{N}$, 
		the term $\sum_{n=1}^{N} u_n^4$ can be expressed~as $\sum_{n=1}^{N} u_n^4 = \frac{d^4 N(N^2-1)(3N^2-7)}{240} \approx \frac{d^4N^5}{80}$.
		As such, the CSV error $\varepsilon_{r,{\rm Tayl}}$ is rewritten as
		\begin{align}
			\varepsilon_{r,{\rm Tayl}} = \frac{\pi\cos^2\hat{\theta}_{\rm E}}{2 \lambda \hat{r}_{\rm E}^2} \frac{N^2d^2}{\sqrt{20}} \Delta r_{\rm E},
		\end{align}
		i.e.,  $\varepsilon_{r,{\rm Tayl}} = \varepsilon_{r,{\rm appr}} $. Similarly, the CSV error in the angle domain obtained via Taylor approximation $\varepsilon_{\theta,{\rm Tayl}}$ can be expressed as
		\begin{align}
			&\varepsilon_{\theta,{\rm Tayl}} = \| \triangledown_{\theta}\mathbf{a}|_{(\hat{\theta}_{\rm E},\hat{r}_{\rm E})}\Delta \theta_{\rm E} \|_{2}  \nonumber \\
			\overset{(d)}{\approx}	& \sqrt{\sum_{n=1}^N \frac{1}{N} \Big( \frac{2\pi u_{n}\cos \hat{\theta}_{\rm E}}{\lambda} \Big)^2} \Delta \theta_{\rm E} =  \frac{2\pi\cos\hat{\theta}_{\rm E}}{\lambda \sqrt{N}}\sqrt{\sum_{n=1}^{N} u_n^2 } \Delta \theta_{\rm E}, \nonumber
		\end{align}
		where $(d)$ since  $\frac{u_{n}^2\sin \hat{\theta}_{{\rm E}} \cos \hat{\theta}_{{\rm E}}}{\hat{r}_{\rm E}}$ is sufficiently  small. Given that ${\sum_{n=1}^{N} u_n^2 } \approx \frac{N^3\lambda^2}{48}$, we obtain $\varepsilon_{\theta,{\rm Tayl}}=\frac{\pi N \cos \hat{\theta}_{\rm E}}{\sqrt{12}} \Delta \theta_{\rm E}=  \varepsilon_{\theta,{\rm appr}} $, thus completing the proof.
		
		\section{}\label{App:partialTaylor}
		By denoting $\Delta \mathbf{a}_{\theta} \!=\! \mathbf{a}({\theta}_{\rm E},\hat{r}_{\rm E}) -  \mathbf{a}(\hat{\theta}_{\rm E},\hat{r}_{\rm E}) - \triangledown_{\theta}\mathbf{a}|_{(\hat{\theta}_{\rm E},\hat{r}_{\rm E})}\Delta \theta_{\rm E}$, the squared-norm of $\Delta \mathbf{a}_{\theta}$ is given by
		\begin{align}
			\|\Delta \mathbf{a}_{\theta}\|_{2}^{2} 	&=  \|\mathbf{a}({\theta}_{\rm E},\hat{r}_{\rm E})-\mathbf{a}(\hat{\theta}_{\rm E},\hat{r}_{\rm E})\|_{2}^{2} + \|\triangledown_{\theta}\mathbf{a}|_{(\hat{\theta}_{\rm E},\hat{r}_{\rm E})}\Delta \theta_{\rm E} \|_{2}^{2} \nonumber \\
			&- 2\mathcal{R} \{\big( \mathbf{a}^H({\theta}_{\rm E},\hat{r}_{\rm E})
			- \mathbf{a}^H(\hat{\theta}_{\rm E},\hat{r}_{\rm E})\big) \triangledown_{\theta}\mathbf{a}|_{(\hat{\theta}_{\rm E},\hat{r}_{\rm E})}\Delta \theta_{\rm E}  \} \nonumber \\
			& = 2 \varepsilon_{\theta,{\rm Tayl}}^{2} - 2\mathcal{R}\{\mathbf{a}^H({\theta}_{\rm E},\hat{r}_{\rm E}) \triangledown_{\theta}\mathbf{a}|_{(\hat{\theta}_{\rm E},\hat{r}_{\rm E})}\Delta \theta_{\rm E}\}. \nonumber
		\end{align}
		Based on the specific  expression of $\triangledown_{\theta}\mathbf{a}$ in~\eqref{Exp:Taylorwrttheta}, we have
		\begin{align}
			&\quad 2\mathcal{R}\{\mathbf{a}^H({\theta}_{\rm E},\hat{r}_{\rm E}) \triangledown_{\theta}\mathbf{a}|_{(\hat{\theta}_{\rm E},\hat{r}_{\rm E})}\Delta \theta_{\rm E}\} \nonumber \\
			&\overset{(e)}{\approx} 2\mathcal{R}\Big\{\frac{\jmath}{N}\sum_{n=1}^{N} e^{\jmath\frac{2\pi}{\lambda}u_{n}(\sin\hat{\theta}_{{\rm E}} -\sin\theta_{{\rm E}} ) } \cdot \frac{2\pi}{\lambda}u_{n} \cos\hat{\theta}_{\rm E} \Delta \theta_{\rm E} \Big\} \nonumber \\
			& \overset{(f)}{\approx} 2\mathcal{R}\Big\{\frac{\jmath}{N}\sum_{n=1}^{N} e^{\jmath\frac{2\pi}{\lambda}u_{n} \cos\hat{\theta}_{\rm E} \Delta \theta_{\rm E} } \cdot \frac{2\pi}{\lambda}u_{n} \cos\hat{\theta}_{\rm E}\Delta \theta_{\rm E} \Big\} \nonumber \\
			& \overset{(g)}{\approx} 2 \sum_{n=1}^{N} 
			\big(\frac{2\pi}{\lambda}u_{n} \cos\hat{\theta}_{\rm E} \Delta \theta_{\rm E} \big)^2/N =   2 \varepsilon_{\theta,{\rm Tayl}}^{2},
		\end{align}
		where $(e)$ holds due to the sufficiently small quadratic terms,
		$(f)$ holds by using $\sin\theta_{\rm E}-\sin\hat{\theta}_{\rm E} \approx \cos \hat{\theta}_{\rm E}  \Delta \theta_{\rm E}$, and $(g)$ holds when 
		$\sin\theta_{\rm E}-\sin\hat{\theta}_{\rm E} \le \frac{1}{2N}$. Based on the above, we obtain $\|\Delta \mathbf{a}_{\theta}\|_{2}^{2} \approx 0$. Similarly, it can be easily shown that  $\|\Delta \mathbf{a}_{r}\|_{2}^{2} \approx 0$, with details omitted for brevity.

		\section{}\label{App:LMI}
		Given arbitrary $\!\boldsymbol{\zeta}_{s}\!$ that satisfies $ \boldsymbol{\zeta}_{s}^T \boldsymbol{\Sigma}_{s}^{-1} \boldsymbol{\zeta}_{s} \le1 $, the constraint~\eqref{Exp:SU_set} holds when 
		$\small	\left[
			\begin{array}{@{}c@{}c@{}}
				{\Gamma} & \mathbf{w}^{H}(\mathbf{a}_{s} + \mathbf{J}_{s} \boldsymbol{\zeta}_{s} ) \\
				(\mathbf{a}_{s} + \mathbf{J}_{s} \boldsymbol{\zeta}_{s})^{H} \mathbf{w} & 1
			\end{array}
			\right] \succeq \mathbf{0} \small$,
		 where $\mathbf{J}_{s} = [\triangledown_{r}\mathbf{a},\triangledown_{\theta}\mathbf{a}]\big|_{( \varphi_{s}, r_{s} )}$ is the gradient matrix,
		$ \boldsymbol{\zeta}_{s} = [\Delta r_{s},  \Delta \varphi_{s}]^T $ is the error vector,  
		and $\boldsymbol{\Sigma}_{s}=\text{diag}(\epsilon_{s}^2,\vartheta_{s}^2) $. By expanding the above LMI, it can be rewritten as
		{\small\begin{align}
			&\quad \left[
			\begin{array}{c@{}c@{}}
				{\Gamma} & \mathbf{w}^{H}\mathbf{a}_{s} \\
				\mathbf{a}_{s}^{H} \mathbf{w} & 1
			\end{array}
			\right] \succeq
			- 	\left[
			\begin{array}{@{}c@{}c@{}}
				0 & \mathbf{w}^{H} \mathbf{J}_{s} \boldsymbol{\zeta}_{s} \\
				( \mathbf{J}_{s} \boldsymbol{\zeta}_{s})^{H} \mathbf{w} & 0
			\end{array}
			\right] \nonumber \\
			= &-\left[\begin{array}{@{}c@{}}
				\mathbf{w}^{H} \triangledown_{r}\mathbf{a}\big|_{( \varphi_{s}, r_{s} )} \\
				0
			\end{array}\right]\Delta r_{s}[0, 1 ]
			-\left[\begin{array}{@{}c@{}}
				\mathbf{w}^{H} \triangledown_{\theta}\mathbf{a}\big|_{( \varphi_{s}, r_{s} )} \\
				0
			\end{array}\right]\Delta \varphi_{s}[0, 1 ]
			\nonumber \\
			&- \left[\begin{array}{@{}c@{}}
				0 \\
				1 
			\end{array}\right] \Delta r_{s} [\triangledown_{r}\mathbf{a}^H\big|_{( \varphi_{s}, r_{s} )}\mathbf{w}, 0] 
			- \left[\begin{array}{@{}c@{}}
				0 \\
				1 
			\end{array}\right] \Delta \varphi_{s} [\triangledown_{\theta}\mathbf{a}^H\big|_{( \varphi_{s}, r_{s} )}\mathbf{w}, 0]. \nonumber
		\end{align}}Considering that $ |\Delta r_{s}| \le \epsilon_{s}$	and $|\Delta \varphi_{s}| \le \vartheta_{s} $, based on the GSD lemma in~\cite{GuiZhou_Robust},  constraint~\eqref{Exp:SU_set} can be rewritten as the LMI in~\eqref{Exp:LMI},
		where ${\lambda}_{s}^{(r)}\ge0$ and $ {\lambda}_{s}^{(\theta)}  \ge 0$ are auxiliary variables,  ${\mathbf{b}}_{s}^{(r)} = [\triangledown_{r}\mathbf{a}^H\big|_{( \varphi_{s}, r_{s} )}\mathbf{w}, 0 ]^H$, and  ${\mathbf{b}}_{s}^{(\theta)} = [\triangledown_{\theta}\mathbf{a}^H\big|_{( \varphi_{s}, r_{s} )}\mathbf{w}, 0 ]^H$, thus completing the proof.

	\end{appendices}

	\bibliographystyle{IEEEtran}
	\bibliography{Secure.bib}
	
\end{document}